\documentclass[12pt]{scrartcl}
	\usepackage[english]{babel}
	\usepackage{amssymb,amsfonts,amsmath, amsthm}
	\usepackage{mathtools}
	\usepackage{xcolor}
	\usepackage{setspace}
    \usepackage{csquotes}
    \usepackage{dsfont}
    \usepackage{sgame}
    \usepackage{tikz}
    \usepackage{hyperref}
    \hypersetup{colorlinks=true, linkcolor=blue, urlcolor=blue, citecolor=blue}

    \usepackage[backend=biber, style=authoryear, sorting=nyt]{biblatex}
    \newcommand{\citet}[1]{\textcite{#1}}
    \newcommand{\citep}[1]{(\cite{#1})}

	\newcommand{\bbE}{\mathbb{E}}

	\newcommand{\bbP}{\mathbb{P}}

	\makeatletter
		\newcommand{\Matrix}[1] 	
		    {\begin{pmatrix}
		      \Matrix@r #1;\@bye;\Matrix@r
		     \end{pmatrix}}
		\def\Matrix@r #1;{\@bye #1\Matrix@z\@bye\Matrix@s #1,\@bye, }%
		\def\Matrix@s #1,{#1\Matrix@t }%
		\def\Matrix@t #1,{\@bye #1\Matrix@y\@bye\@firstofone {&#1}\Matrix@t}%
		\def\Matrix@y #1\Matrix@t{\\ \Matrix@r }%
		\def\Matrix@z #1\Matrix@r {}
		\def\@bye  #1\@bye {}
	\makeatother
	
		\newtheorem{claim}{Claim}
	\newtheorem{corollary}{Corollary}

	\newtheorem{definition}{Definition}
	\newtheorem{lemma}{Lemma}
	\newtheorem{proposition}{Proposition}

	\DeclareMathOperator{\var}{Var}
	
	\DeclareMathOperator*{\argmin}{{arg\,min}}
	\DeclareMathOperator*{\argmax}{{arg\,max}}

\newcommand{\tht}{\theta}
\newcommand{\ve}{\varepsilon}

	\newcommand{\given}{\ifnum\currentgrouptype=16 \;\middle|\;\else\mid\fi}

\newcommand{\ind}{\mathds{1}}

\title{Cross-Validation Equilibrium%
\thanks{We acknowledge financial support from UKRI Frontier Research Grant no. EP/Y033361/1. We thank participants at the 5th Southeast Theory Festival and the referees at EC2026, for helpful comments.}}

\subtitle{}

\author{
Ran Spiegler%
\thanks{Tel Aviv University and University College London, \href{mailto:rani@tauex.tau.ac.il}{rani@tauex.tau.ac.il}.}
\ and
Stephan Waizmann%
\thanks{University College London, \href{mailto:stephanwaizmann@gmail.com}{stephanwaizmann@gmail.com}.}
}
	
\begin{document}

\onehalfspacing

\maketitle

\begin{abstract}
	   We study strategic interaction when players delegate belief formation to predictive machine learning (ML). In a static Bayesian game, each player’s ML agent predicts a payoff-relevant outcome variable as a function of the player’s type. The ML agent's training sample is endogenous: it is drawn from the outcome distribution generated by players’ ML-guided behavior. In Cross-Validation Equilibrium (CVE), each player's ML agent selects a predictive model to minimize expected out-of-sample squared error, given its realized training sample; and each player best-replies to the belief generated by the model her ML agent selected. We analyze CVE and relate it to other equilibrium concepts. We apply CVE to jury voting, speculative betting, and games with linear-quadratic payoffs. E.g., in a team-effort game, endogenous model selection can give rise to multiple equilibria.

\end{abstract}

\newpage

\section{Introduction}

The use of machine learning (ML) algorithms to form predictive beliefs in economic settings has become ubiquitous. Credit agencies employ ML to evaluate borrowers' creditworthiness. Judges make use of pretrial risk assessment algorithms to predict whether a person accused of a crime will fail to show up at court or engage in criminal activity if they are bailed out. In these and other cases, ML algorithms use a random training sample to generate conditional probabilistic predictions; the ultimate decision power rests with the human decision-maker who employs these tools.

This paper addresses the following question: What are the effects of predictive ML when its users are players in a static game with possibly incomplete information? In this context, an individual player tasks an \textquotedblleft ML agent\textquotedblright\ with constructing a predictive model (selected on the basis of a random training sample) that outputs a belief over payoff-relevant variables (e.g., an opponent's action) as a function of the player's signal. The random training sample used by the ML agent is drawn from an endogenous distribution induced by all players' responses to their own ML-generated beliefs. We are interested in the implications of this two-way linkage between players' behavior and ML-generated predictive models.

We develop a game-theoretic equilibrium concept, dubbed \textit{Cross-Validation Equilibrium (CVE)}, in which players' equilibrium beliefs are formed according to a stylized model of predictive ML. Our concept captures the idea that predictive ML has access to a noisy training sample --- drawn from the equilibrium distribution --- and uses it to construct a predictive model of a payoff-relevant variable. The predictive model is selected from some feasible set of models according to its \textit{out-of-sample predictive performance}, measured by the \textit{expected mean squared prediction error}. Thus, while predictive models are estimated against the training sample, their selection is based on the out-of-sample predictive success of each estimated model. This model-selection criterion captures in stylized form the notion of \textit{cross validation} that characterizes many ML tools (e.g., see Chapter 7 in \citet{HastieEtAl2009}).

Since players' endogenous training samples are random, so is the predictive model they end up adopting, and therefore so is their resulting behavior. Thus, CVE involves randomness in players' subjective models, the beliefs these models generate, and the behavior they induce. Nash equilibrium is a limit case of our concept, corresponding to CVE with perfectly noiseless samples, which give rise to correct models and therefore correct beliefs. Another limit case is Analogy Based Expectation Equilibrium \citet{Jehiel2005}, which corresponds to a CVE in which each player's model is exogenously given and training samples are noiseless. 

For a simple example that illustrates our modeling approach, consider a standard symmetric two-firm Cournot competition with linear demand and a random, commonly observed demand parameter $\theta $. Firm $i$'s ML agent aims to predict firm $j$'s expected quantity as a function of $\theta $. Imagine that for each value of $\theta $, firm $i$'s ML agent observes firm $j$'s expected quantity conditional on $\theta $ plus independent Gaussian noise. This noisy observation is $not$ a standard game-theoretic signal; it is an observation of the \textit{equilibrium} link between $\theta $ and firm $j$'s behavior. The profile of noisy measurements constitutes the \textquotedblleft training sample\textquotedblright\ at the disposal of firm $i$'s ML agent.

Firm $i$'s ML agent selects a \textit{predictive model}, which (following \citet{Jehiel2005}) is a \textit{partition} of the set of parameter values. For each partition cell, the prediction is the weighted training-sample average of firm $j$'s quantity conditional on the cell. The ML agent chooses a partition that minimizes the expected \textit{out-of-sample squared prediction error}. This procedure naturally involves a bias-variance trade-off: For some training-sample noise realizations, the ML agent will adopt a coarse partition that bundles together different values of $\theta $, even if they are associated with different behavior by firm $j$. In CVE, firm $i$ best-replies to the prediction induced by the selected partition.

We establish that a CVE exists in finite environments. Moreover, a CVE converges to an ABEE with respect to the finest feasible partition as the training-sample noise vanishes.  We also highlight basic properties of CVE. First, the expected training-sample noise conditional on any partition cell is zero. This unbiasedness property, which improves the concept's tractability, is not obvious a priori, as the partition itself is endogenously selected. It arises from our criterion for out-of-sample predictive success, and from the assumption of symmetric noise. Second, the ML agent avoids over-fitting, in the sense that if the predicted variable is measurable with respect to some partition, the ML agent will (almost) never adopt a finer one. 

We then apply CVE to various games, to explore its economic applications. The concept is especially tractable in games with linear-quadratic payoffs. For example, in the above-mentioned Cournot competition example (when the only feasible partitions are the fully fine and fully coarse ones), we show that the volatility of quantities and prices in the unique symmetric CVE rises with the variance of the training-sample noise. We also show that unlike Nash equilibrium, asymmetric CVE are possible in this game. Another example involves a team effort problem, where unlike the Cournot example, the game exhibits strategic complementarities. As a result, and despite the linear best-replies that imply a unique Nash equilibrium, there can be \textit{multiple} symmetric CVE. We also establish a sense in which CVE in the team effort game is discontinuously different from Nash equilibrium, even when the training sample noise is arbitrarily small. We also apply CVE to a two-player version of \citet{FeddersenPesendorfer1998}'s jury model, as well as a simple two-player game of speculative trade, where we show that when one player is restricted to use the correct model, this player may make an expected loss to her opponent in CVE.\bigskip 

The rest of the paper is organized as follows. Section \ref{sec: model} presents the model and defines Cross-Validation Equilibrium. The general results are contained in Section \ref{sec: results}. In Section \ref{sec: applications}, we apply CVE to examples. Proofs of the main results are in the \hyperref[sec: proofs]{Appendix}.

\section{The Model}
\label{sec: model}

Consider the following, slightly non-standard formulation of a Bayesian game. The set of players is $N=\{1,...,n\}$. Let $\Theta$ be a set of states of Nature. For each player $i\in N$, let $A_{i}$ be the player's action set, and let $T_{i}$ be the set of types for this player. Denote $T=\times _{i\in N}T_{i}$ and $A=\times _{i\in N}A_{i}$. Let $\mu \in \Delta(\Theta \times T)$ be an objective prior distribution over the exogenous variables. A strategy for player $i$ is a function $\sigma_{i}:T_{i}\rightarrow \Delta (A_{i})$.

So far, the exposition is standard. We slightly deviate from the norm in our definition of payoffs. Let $X_{i}=\mathbb{R}$ be a set of \textit{payoff-relevant outcomes} for player $i$. Let $f_{i}:\Theta \times A_{-i}\rightarrow X_{i}$ be a function that maps states and profiles of actions by player $i$'s opponents to her payoff-relevant outcomes. Let $u_{i}:A_{i}\times T_{i}\times X_{i}\rightarrow \mathbb{R}$ be player $i$'s \textit{continuous} vNM utility function, which is in particular \textit{linear} in $x_{i}$. The scalar variable $x_{i}$ is a \textquotedblleft sufficient statistic\textquotedblright\ that pins down player $i$'s payoff given her action and type. This variable is what player $i$'s ML-based belief-formation procedure will aim to predict.

The following examples illustrate this non-standard aspect of the formalism. In both examples, we leave the information structure given by $\mu $ unspecified, because it is immaterial for the illustration.\bigskip

\noindent \textit{Example 2.1} (mixtures in binary-action games). Let $n=2$ and $A_{1}=A_{2}=[0,1]$, where $a_{i}$ represents the probability that player $i$ chooses a particular action in a binary-action game. In this case, let $f_{i}(\theta ,a_{j})=a_{j}$. When $u_{i}$ is player $i$'s expected utility in the binary-action game, linearity of $u_{i}$ in $x_{i}$ follows automatically.\bigskip

\noindent \textit{Example 2.2 (a linear Cournot game)}. Let $n=2$ and $A_{1}=A_{2}=\mathbb{R}$. Let $\theta $ be a demand parameter, such that $u_{i}=a_{i}[\theta-a_{i}-a_{j}]$. (Allow the price $\theta -a_{i}-a_{j}$ to take negative values.) Consider two alternative specifications of $x_{i}$. First, let $f_{i}(\theta ,a_{j})=a_{j}$ --- i.e., the sufficient statistic for $u_{i}$ is player $j$'s quantity. Second, let $f_{i}(\theta ,a_{j})=\theta -a_{j}$ --- i.e., the sufficient statistic for $u_{i}$ is the residual demand that player $i$ faces.\bigskip

The combination of $\mu $ and a strategy profile $\sigma =(\sigma_{1},...,\sigma _{2})$ induces a joint distribution over $\Theta \times T\times A$, 
\begin{equation}
p^{\sigma }(\theta ,t,a)=\mu (\theta ,t)\cdot \prod_{i\in N}\sigma
_{i}(a_{i}\mid t_{i})  \label{eq: p-sigma}
\end{equation}
For convenience, we also use the notation $p^{\sigma }$ for marginal and conditional distributions derived from this joint distribution. Let
\begin{equation}
\bar{x}_{i}^{\sigma }(t_{i})=\sum_{\theta ,a_{-i}}p^{\sigma }(\theta
,a_{-i}\mid t_{i})f_{i}(\theta ,a_{-i})  \label{eq: x-bar-sigma}
\end{equation}
be the expectation of $x_{i}$ conditional on player $i$'s type, according to $p^{\sigma }$. Note that $p^{\sigma }(\theta ,a_{-i}\mid t_{i})$ --- and therefore also $\bar{x}_{i}^{\sigma }$ --- only depend on the profile $\sigma _{-i}$; they are invariant to player $i$'s own mixed strategy $\sigma_{i}$.

We now describe how player $i$ forms a prediction of her payoff-relevant outcome given her type. Fix $\sigma _{-i}$. This means that $\bar{x}_{i}^{\sigma }(t_{i})$ is also held fixed for every $t_{i}$. An \textit{ML agent} at player $i$'s service has access to a \textit{training sample}, described by
\begin{equation}
y^{t_{i}}=\bar{x}_{i}^{\sigma }(t_{i})+\varepsilon _{t_{i}}
\label{eq: training-sample}
\end{equation}
for every $t_{i}\in T$, where $\varepsilon _{t_{i}}$ is $i.i.d$ according to a distribution $\upsilon $ that is distributed symmetrically around zero.  Since $\bar{x}_{i}^{\sigma }(t_{i})$ is calculated with respect to the conditional distribution $p^{\sigma }(\theta ,a_{-i}\mid t_{i})$, the training sample can be viewed as a \textquotedblleft representative sample\textquotedblright\, in the spirit of \citet{DanenbergSpiegler2025}. 

For expositional convenience, we will assume here that $\upsilon $ has finite support (in applications, we employ continuous noise as well). Thus, given $\bar{x}_{i}^{\sigma }=(\bar{x}_{i}^{\sigma}(t_{i}))_{t_{i}\in T_{i}}$, we identify the sample with $\varepsilon
_{i}=(\varepsilon _{t_{i}})_{t_{i}\in T_{i}}$. For any $M\subseteq T_{i}$, let 
\begin{equation}
\hat{p}^{\varepsilon _{i}}(x_{i}\mid M)=\sum_{t_{i}\in M}\mu (t_{i}\mid
t_{i}\in M)\cdot \mathbf{1}\{\bar{x}_{i}^{\sigma }(t_{i})+\varepsilon
_{t_{i}}=x_{i}\}  \label{eq: empirical-distribution}
\end{equation}%
be the frequency of $x_{i}$ in the sub-sample $\{\varepsilon_{t_{i}}\}_{t_{i}\in M}$, weighted according to the distribution of $t_{i}$ conditional on $t_{i}\in M$. Likewise, let
\begin{equation}
\hat{x}^{\varepsilon _{i}}(M)=\sum_{t_{i}\in M}\mu (t_{i}\mid t_{i}\in M)(%
\bar{x}_{i}^{\sigma }(t_{i})+\varepsilon _{t_{i}})  \label{x hat}
\end{equation}%
be the weighted average value of $x_{i}$ in the sub-sample.

Given the sample $\varepsilon _{i}$, the ML agent chooses a partition $\Pi_{i}(\varepsilon _{i})$ of $T_{i}$ from some family $\mathcal{P}_{i}$ of feasible partitions. Let $\pi (t_{i})\in \Pi _{i}(\varepsilon _{i})$ be the partition cell that includes $t_{i}$. The ML agent evaluates the partition's out-of-sample predictive performance with a \textit{validation sample}, which observes $\bar{x}_{i}^{\sigma }(t_{i})$ with $i.i.d$ zero-mean noise $\eta _{t_{i}}$ for each $t_{i}$. The ML agent selects the partition with the best average out-of-sample predictive performance. This partition generates a belief, and player $i$ best-replies to this belief. In what follows, $g_{i}:T_{i}\rightarrow A_{i}$ denotes a pure ex-ante strategy for player $i$

\begin{definition}
    \label{def: compatible}
A pair $(\Pi _{i},g_{i})$ is compatible with $\sigma _{-i}$ and $\varepsilon
_{i}$ if%
\begin{equation}
\Pi _{i}\in \arg \min_{\Pi _{i}^{\prime }\in \mathcal{P}_{i}}\bbE_{\eta}\sum_{t_{i}\in T_{i}}\mu (t_{i})[\hat{x}^{\varepsilon _{i}}(\pi^{\prime }(t_{i}))-(\bar{x}_{i}^{\sigma }(t_{i})+\eta _{t_{i}})]^{2} \label{eq: optimal-partition}
\end{equation}%
and, for every $t_{i}\in T_{i}$,%
\begin{equation}
g_{i}(t_{i})\in \arg \max_{a_{i}\in A_{i}}\sum_{x_{i}\in X_{i}}\hat{p}^{\varepsilon _{i}}(x_{i}\mid \pi (t_{i})) \, u_{i}(a_{i},t_{i},x_{i}) \label{eq: optimal-action}
\end{equation}
\end{definition}

The first condition requires that given the training sample $\varepsilon_{i} $, the partition $\Pi_{i}$ minimizes the expected out-of-sample mean squared prediction error, with respect to the true expected payoff-relevant outcomes in the validation sample given by $\bar{x}_{i}^{\sigma }$. The expectation with respect to $\eta $ means that the ML agent evaluates a partition's predictive performance across many validation samples. The second condition is that player $i$'s ex-ante pure strategy prescribes a best-reply to the belief induced by the partition $\Pi $. Note that since $u_{i}$ is assumed to be linear in $x_{i}$, we can rewrite player $i$'s expected payoff in (\ref{eq: optimal-action}) as
\begin{equation}
 u_{i}(a_{i},t_{i},\sum_{x_{i}\in X_{i}}\hat{p}^{\varepsilon _{i}}(x_{i}\mid \pi (t_{i}))x_{i} \,).
\end{equation}

We are now able to close the model and define when a strategy profile constitutes an equilibrium.\bigskip

\begin{definition}
    \label{def: CVE}
A strategy profile $\sigma $ is a Cross-Validation Equilibrium (CVE) if for every $i\in N$ and every training  sample realization $\varepsilon _{i}$, there is a mixture $\beta _{\varepsilon _{i}}$ over pairs $(\Pi _{i},g_{i})$ that are compatible with $\sigma _{-i}$ and $\varepsilon_{i}$, such that for every $t_{i}\in T_{i}$,
\begin{equation}
\sigma _{i}(a_{i}\mid t_{i})=\sum_{\varepsilon _{i}}\upsilon (\varepsilon_{i})\sum_{(\Pi _{i},g_{i})\mid g_{i}(t_{i})=a_{i}}\beta _{\varepsilon_{i}}(\Pi _{i},g_{i})  \label{eq: strategy}
\end{equation}
\end{definition}

Thus, in CVE, a player's mixed strategy arises from her ML-based belief-formation procedure. The player's ML agent has access to a random sample of the type-dependent payoff-relevant outcome $\bar{x}_{i}^{\sigma }$, forms an out-of-sample-optimal partition of $T_{i}$ that induces a type-dependent belief over payoff-relevant outcome. The player best-replies to this belief. Because the player's sample is random, so is her best-replying behavior, and this randomness is described by the player's mixed strategy.

\paragraph{A simple illustrative example}
Consider the following symmetric, binary-action, two-person game with complete information. Players' action set is $\{0,1\}$. The set of states of Nature is $\{0,1\}$. The two states are equally likely. The state is commonly known so that $t_i=\tht$ with probability $1$. When $\theta =0$, player $i$'s payoff is $-a_{i}$. When $\theta =1$, player $i$'s payoff is $a_{i}(\frac{4}{3}a_{j}-1)$. In the notation of the model,
\[u_i(a_i, \theta, x)=
    \begin{cases}
        - a_i & \tht=0\\
        a_i (\frac{4}{3} x -1) & \tht=1
    \end{cases}
\]
and $f_i(\tht, a_j)=a_j$. Consequently, we can interpret $\overline{x}^\sigma_i(\tht)=\sigma_j(1\mid\tht)$ as the probability with which player $j$ plays action $a_j=1$ when the state is $\tht$.

Clearly, the game has a Nash equilibrium in which both players play $a=\theta $.\footnote{In state $\tht=0$, this is the only Nash equilibrium. In state $\tht=1$, there are two more Nash equilibria: $a_1=a_2=0$ and $\sigma_1(1)=\sigma_2(1)=3/4$.} This is the efficient outcome. Note that player $i$'s best-reply in state $\tht=1$ is $a_{i}=1$ if she assigns probability above $\frac{3}{4}$ to player $j$ choosing $a_{j}=1$, i.e., if $\overline{x}^\sigma_i(1)\geq \frac{3}{4}$.

Assume the training sample noise distribution is uniform over $\{-d,d\}$, where $d>\frac{1}{2}$. $V$ to denote the variance of the validation sample noise. Use We solve for symmetric CVE. Denote $\overline{x}_\tht=\overline{x}^\sigma_i(\tht)$.

In state $\theta =0$, each player $i$ strictly prefers playing $a=0$, independently of her belief regarding $x_{i}$. Consequently, $\overline{x}_0=0$. Therefore, we only need to derive the probability that each player chooses $a=1$ in state $\theta =1$. This probability equals $\overline{x}_1 $. 

The expected out-of-sample prediction error when adopting the fine partition ${{0},{1}}$ is $\frac{1}{2}(\ve_0^2+\ve_1^2) +V$, whereas the out-of-sample prediction error from adopting the coarse partition ${{0,1}}$ is $(\frac{\overline{x}_1}{2} +\frac{\ve_0+\ve_1}{2})^2 +V$.
When $(\varepsilon _{1}-\varepsilon _{0})^{2}\geq (\bar{x}_{1}-\bar{x}_{0})^{2}=\overline{x}_1^{2}$, player $i$'s ML agent adopts the coarse partition, and the player's estimate of $x_{i}$ is $\frac{1}{2}(\overline{x}_1 +\varepsilon_{0}+\varepsilon _{1})$. When $(\varepsilon _{1}-\varepsilon_{0})^{2}<\overline{x}_1 ^{2}$, player $i$'s ML agent adopts the fine partition, and the player's estimate of $x_{i}$ is $\overline{x}_1 +\varepsilon _{1}$. The player will play $a=1$ only if her estimate of $x_{i}$ is weakly above $\frac{3}{4}$.

Let us now use these conditions, combined with the noise distribution, to derive $\overline{x}_1 $. When $\varepsilon _{1}=\varepsilon _{0}$, the ML agent necessarily adopts the fine partition. Thus, the coarse partition is adopted only if $\varepsilon _{1}=-\varepsilon _{0}$. But in this case, the player
plays $a=1$ only if
\[
\frac{1}{2}(\overline{x}_1 +\varepsilon _{0}+\varepsilon _{1})=\frac{1}{2}\overline{x}_1 >\frac{3}{4}
\]
which this is impossible. It follows that $\overline{x}_1$ is bounded from above by the probability that an ML agent adopts the fine partition.

When a player's ML agent adopts the fine partition, she plays $a=1$ only if $\overline{x}_1 +\varepsilon _{1}\geq \frac{3}{4}$. Since $d>\frac{1}{2}$, it follows that $\varepsilon _{1}=d$ is necessary for playing $a=1$. Thus, $\overline{x}_1 \leq \frac{1}{2}$. If $\varepsilon _{0}=-d$, then 
\[
(\varepsilon _{1}-\varepsilon _{0})^{2}=4d^{2}>\frac{1}{4}>\overline{x}_1 ^{2}
\]
such that the ML agent adopts the coarse partition. It follows that players choose $a=1$ only if $\varepsilon _{1}=\varepsilon _{0}=d$, such that $\overline{x}_1 \leq \frac{1}{4}$. 
It follows that there exists a  CVE with $\overline{x}_1=\frac{1}{4}$.

\paragraph{Discussion}
We conclude this section with a discussion of some of our modeling choices.

In our framework, the player's belief formation is entirely delegated to the ML agent. The player herself is not assumed to have a thorough understanding of the game she is playing. Indeed, she need not perceive the situation as a game; she is interested in quantifying the mapping from $t_{i}$ to $x_{i}$. In this sense, our model follows the approach of \citet{OsborneRubinstein1998}. But while in that paper, players relied on naive sampling to evaluate such mappings, in our model the player employs an ML agent to quantify the mapping. The player does understand that $x_{i}$ should be attributed to $t_{i}$ only, i.e., that it is independent of $a_{i}$. She may also have certain preconceptions that limit the set of partitions she allows the ML agent to explore. But that is the only extent of her understanding of the situation, as far as belief formation is concerned.

The space of models that a player's ML agent explores is a set of partitions of the player's type space. Of course, one could imagine other classes of predictive models; our use of partitions can be regarded as a \textquotedblleft proof of concept.\textquotedblright\  It has two advantages. First, it is non-parametric and can be applied to any Bayesian game (whereas parametric predictive models would require knowledge of the specifics of the Bayesian game). Second, it enables us to form a link to the concept of Analogy-Based Expectations Equilibrium (ABEE, due to \citet{Jehiel2005}), as we will see in Section \ref{sec: results}.

Note that the domain of partitions that the ML agent considers may be restricted. There are two motivations for this feature of our framework. First, the player may have a-priori theoretical reasons to impose structure on the set of predictive models. For example, when $T_{i}$ is two-dimensional, i.e., it is the product of two sets, the player may believe that this product structure should be preserved in any model that predicts $x_{i}$. Second, we rely on restricted domains in examples for tractability, e.g. by assuming that the only feasible partitions are the fully coarse and fully fine ones.
Tractability considerations also lie behind our treatment of the training sample noise --- specifically, our assumption that the noise is additive and independent of the true expected value of the predicted payoff-relevant outcome.

Our assumption that $x_{i}$ is a scalar is made for expositional simplicity. A natural extension would be that $x \in \mathbb{R}^n$, and that out-of-sample predictive fit would be evaluated by the sum of squared prediction errors across dimensions. Our assumption that player $i$'s payoff is linear in $x_{i}$ should be viewed in this context. Linearity implies that when player $i$ has uncertainty about $x_{i}$, the mean with respect to her belief is a sufficient statistic for predicting her payoff --- i.e., she only needs to predict a scalar variable in order to predict her payoff. If we relaxed linearity, the restriction to a scalar $x_{i}$ would be less well-founded. Without linearity, the tight link between CVE and ABEE would break down.

CVE is based on a criterion for out-of-sample predictive success that involves taking an expectation over the validation sample noise. One intepretation is that the ML agent's estimated model is pitted against multiple (independent) validation samples, and the model is evaluated according to its average predictive success across these samples. However, our motivation behind this assumption is tractability: as we will see below, it makes the analysis of CVE more tractable. Note that the exact shape of the validation sample noise is immaterial for our purposes. 

\section{General Results}
\label{sec: results}

In this section, we present results that hold for all finite games. We show that a CVE exists. Furthermore, we consider convergence of CVE as the training sample becomes perfect. We conclude with two results that illustrate properties of the out-of-sample prediction criterion. All proofs are contained in the \hyperref[sec: proofs]{Appendix}.


Say that the game is finite if each player's action set, type space, the set of nature, and the support of the training sample is finite.\bigskip

\begin{proposition}
    \label{prop: existence}
    In every finite game, a CVE exists.\\
\end{proposition}

The proof applies a standard fixed-point argument to a correspondence $\Psi$ that maps, for each realization of the training sample, a distribution over partitions $\lambda$ and a profile of  mixed strategies $\sigma$ into itself. The correspondence returns the partitions that minimize the out-of-sample prediction error, given the training sample and the strategy profile $\sigma$, as well as the mixed strategies that are compatible with the distribution over partitions $\lambda$ and the strategy profile $\sigma$. A fixed point of $\Psi$ corresponds to a CVE. 
\bigskip


Next, we consider the convergence properties of CVE as the training sample becomes perfect. We first recall the definition of an Analogy Based Expectations Equilibrium, adapted to the notation in our setting.\footnote{Note that we use linearity of $u_i$ in $x$ here.}

\begin{definition}{ABEE, \citet{Jehiel2005}}\\
    \label{def: ABEE}
    $\sigma$ is an \emph{Analogy Based Expectations Equilibrium (ABEE)} with respect to the partitions $\{\Pi_i\}$ if for each player $i$ and type $t_i$, for all $a_i$ with $\sigma_i(a_i|t_i)>0,$
    \begin{align*}
        u_i\left(a_i, t_i, \sum_{t_i^\prime\in \pi(t_i)} \mu(t_i^\prime\mid t_i^\prime \in \pi(t_i)) \overline{x}_i^\sigma(t_i^\prime)\right) 
        \geq u_i\left(a_i^\prime, t_i, \sum_{t_i^\prime\in \pi(t_i)} \mu(t_i^\prime\mid t_i^\prime \in \pi(t_i)) \overline{x}_i^\sigma(t_i^\prime)\right) 
    \end{align*}
    holds for all $ a_i^\prime \in A_i$.
\end{definition}

Say that player $i$'s set of feasible partitions $\mathcal{P}_i$ contains a \emph{finest partition} if there is $\Pi\in \mathcal{P}_i$ such that $\Pi$ refines each $\Pi^\prime \in \mathcal{P}_i$. Denote the finest partition by $\Pi_i^{\textrm{finest}}$. 

For each $n$, let $v^n$ be a distribution that has finite support and is symmetric around $0$. Assume that $\var_{\ve \sim v^n}(\ve)\to 0$ as $n\to \infty$. Let $(\ve^n)$ be a sequence of independent random variables where $\ve^n \sim v^n$. Let $\sigma^n$ be a CVE when the training sample is distributed according to $v^n$.

\begin{proposition}
    \label{prop: convergence-to-ABEE}
    Assume the game is finite. Assume each player's set of feasible partitions contains a finest partition.
  If the limit $\sigma=\lim_{n\to \infty} \sigma^n$ exists, then $\sigma$ is an ABEE with respect to each player's finest partition $\Pi_i^{\textrm{finest}}$.\\
\end{proposition}

    Proposition \ref{prop: convergence-to-ABEE} states that an accumulation point of CVEs  as the noise of the training sample vanishes is an ABEE with respect to the finest partition. Consequently, if each player's finest partition is maximally fine (such that each type is in different cell), then CVE converges to a Nash equilibrium when the training sample noise vanishes.

    The key step in the proof of Proposition \ref{prop: convergence-to-ABEE} is to show that each player's belief about the sufficient statistic converges to the belief a player holds when \textit{(i)} she adopts the finest partition, and \textit{(ii)} there is no noise in the belief formation process. While intuitive, this property is not a priori obvious as, for any non-degenerate training sample noise, the partition chosen in a CVE need not coincide with the finest partition. Moreover, the partitions chosen by each player need not converge to the finest partition. In fact, such a convergence fails, in general.\footnote{See Corollary \ref{cor: partition-independent-types}.} Nonetheless, beliefs converge; see Lemma \ref{lemma: belief-convergence}. 
\bigskip


Next, we make an observation regarding which partitions are chosen in CVE. The next Proposition states that the chosen partition does not condition on irrelevant variables.

\begin{proposition}
    \label{prop: redundant-partition}
    Let $\Pi$ and $\Pi^\prime$ be two feasible partitions for player $i$, and assume that $\Pi^\prime$ refines $\Pi$. Suppose that $\sigma_{-i}$ is such that $\overline{x}_i^\sigma(\cdot)$ is measurable with respect to $\Pi$.\footnote{I.e., $\overline{x}_i^\sigma(\cdot)$ is constant on each cell of $\Pi$.} Then, the expected out-of-sample prediction error from partition $\Pi$ is weakly lower than from partition $\Pi^\prime$ for any realization of the training sample noise $\ve_i$.\\
\end{proposition}

\begin{corollary}
    \label{cor: partition-independent-types}
    Assume that for each player $i$, $\mathcal{P}_i$ includes the coarsest partition $\Pi_i^{\textrm{coarse}}=\{T_i\}$. 
    
    If types are independently distributed across players and players have private values (i.e., $f(\theta, a_{-i})$ does not depend on $\theta$), then every CVE is equivalent to a CVE in which each player is forced to choose the coarsest partition.\\
\end{corollary}

The intuition behind Proposition \ref{prop: redundant-partition} is the following. When the expected sufficient statistic $\overline{x}_i^\sigma$ is measurable with respect to $\Pi$, the prediction error arises solely from noise in the training sample. It is then never strictly optimal to choose a finer partition $\Pi^\prime$, as this would result in a higher variance by de-pooling the noise in the training sample. The proof is a straightforward application of Jensen's inequality.

Moreover, the following is also true: if the finer partition $\Pi^\prime$ is chosen, the beliefs it generates are the same as those generated by the partition $\Pi^\prime$. We remark that the finer partition $\Pi^\prime$ is chosen with probability $0$ if the training sample noise has an atomless distribution.

Proposition \ref{prop: redundant-partition} does \textit{not} imply that the out-of-sample-fit criterion selects a partition $\Pi$ such that the sufficient statistic $\overline{x}_i^\sigma(t_i)$ is measurable with respect to $\Pi$. To put it differently: even if $\overline{x}_i^\sigma(t_i)=\overline{x}_i^\sigma(t_i^\prime)$, the types $t_i$ and $t_i^\prime$ need not belong to the same cell. A counterexample is the following. Suppose that there are three types that occur with probability $1/3$. The vector of sufficient statistics is given by $\overline{x}_i^\sigma=(0, 0, 1)$ and the realized training sample noise is $\ve_i=(0, -1, 1)$. Then, the partition that uniquely minimizes the out-of-sample prediction error, $\Pi=\{\{t^1\}, \{t^2, t^3\}\}$, puts types $t^1$ and $t^2$ in different cells. Intuitively, the prediction error can be reduced by introducing some bias if this pools extreme realizations of the training sample noise.


We next make an observation that focuses on the partitions $\Pi_{i}(\varepsilon _{i})$ that are part of a compatible response to $\sigma_{-i}$ and $\varepsilon _{i}$. While the training sample noise has mean zero, it is not obvious a priori that the expected noise conditional on a particular selected partition is also zero. For example, in \citet{EliazSpiegler2019}, the crucial behavioral effects arise precisely because the noise distribution conditional on model selection can be biased. However, the following Proposition establishes that thanks to the validation sample noise and the validation criterion, this type of bias does not exist in our model.

\begin{proposition}
\label{prop: unbiased-noise}
    For every $t_{i}\in T_{i}$ and any partition $\Pi _{i}\in \mathcal{P}_{i}$, the expectation of $\varepsilon _{t_{i}}$ conditional on the event that $\Pi_{i}$ satisfies \eqref{eq: optimal-partition} is zero.\\
\end{proposition}

We emphasize that this result is due to the partition being selected through cross-validation, via \textit{expected} out-of-sample predictive loss with a quadratic loss-function. For other loss functions, the result would fail. Moreover, the assumption that the training-sample noise is symmetrically distributed around $0$ is essential for this result.\bigskip

We conclude this section with an observation regarding partition choice, when only the maximally fine and the maximally coarse partitions are feasible. Let $\var(\ve_i)$ be the variance of the realized training sample, i.e., the variance of the random variable that takes the value $\ve_{t_i}$ with probability $\mu(t_i)$. Similarly, let $\var(\overline{x}_i^\sigma)$ be the variance of the average sufficient statistic.

\begin{proposition}
    \label{prop: fine-vs-coarse-partition}
        Assume the only feasible partitions are the maximally fine partition $\Pi^{\textrm{fine}}=\{\{t_i\}_{t_i\in T_i}\}$ and the maximally coarse partition $\Pi^{\textrm{coarse}}=\{ T_i\}$. Given a realization of the training sample $\ve_i$, the fine partition is adopted if
        \begin{align}
\var(\ve_i)<\var(\overline{x}^\sigma_i) \label{eq: fine-vs-coarse-partition}
        \end{align}
\end{proposition}

Proposition \ref{prop: fine-vs-coarse-partition} provides a tractable partition-selection criterion: if the variance of the realized training sample is less than the variance of the sufficient statistic, the fully fine partition is favored over the fully coarse partition.

\section{Applications}
\label{sec: applications}

In this section we apply CVE to a variety of games, and illustrate the
qualitative effects that the equilibrium concept generates.

\subsection{Quadratic-Payoff Games}
\label{sec: quadratic-games}

When players' action set is the set of real numbers and their payoffs are quadratic in their own action, CVE becomes especially tractable, because each player $i$'s best-reply becomes linear in her point estimate of $x_{i}$.

Consider a two-person game with complete information and an arbitrarily large (finite) set of states, in which players' action set is $\mathbb{R}$, $t_{1}=t_{2}=\theta $, and the payoff function takes the quadratic form:
\begin{equation}
u_{i}(t_{i},a_{i},x_{i})=a_{i}(\gamma \theta +\beta x_{i})-\frac{1}{2}a_{i}^{2}  \label{eq: quadratic-u}
\end{equation}%
where $x_{i}=a_{j}$. Normalize $Var(\theta )=1$. The only feasible partitions are the fully coarse and fully fine ones. Finally, assume $\varepsilon _{\theta }\sim N(0,\sigma ^{2})$.

Clearly, the symmetry of the game and the linearity of players' best-reply functions implies that the game's unique Nash equilibrium is symmetric.

We proceed to characterize CVE. First, when the mean of player $i$'s belief over $x_{i}$ in state $\theta $ is $\hat{x}_{i}(\theta )$, her best-reply is
\begin{equation}
a_{i}=\gamma \theta +\beta \hat{x}_{i}(\theta )  \label{eq: quadratic-best-reply}
\end{equation}
Consider some partition that player $i$'s ML agent adopts, such that $\pi(\theta )$ is the partition cell that includes $\theta $. By definition,%
\[\hat{x}_{i}(\theta )=\sum_{\theta ^{\prime }\in \pi (\theta )}\mu (\theta^{\prime })(\bar{x}_{i}(\theta ^{\prime })+\varepsilon _{\theta ^{\prime }}). 
\]
By Proposition \ref{prop: unbiased-noise}, the expectation of $\varepsilon _{\theta ^{\prime }}$ conditional on $\theta^{\prime }\in \pi (\theta )$ is $zero$, regardless of the partition that was selected.

In CVE, the expectation of player $i$'s best-reply in $\theta $ (which itself is given by \eqref{eq: quadratic-best-reply}) coincides with $\bar{x}%
_{j}(\theta )$. Denote 
\[
x^{\ast }_{i}=\sum_{\theta ^{\prime }}\mu (\theta ^{\prime })\bar{x}_{i}(\theta^{\prime })\qquad \qquad \theta ^{\ast }=\sum_{\theta ^{\prime }}\mu (\theta^{\prime })\theta^{\prime} 
\]
Let $\lambda_{i}$ be the equilibrium probability with which player $i$'s ML agent adopts the fine partition. Then, since $\bar{x}_{i}(\theta)$ coincides with the expectation of $a_{j}$ in state $\theta$, we can use (\ref{eq: quadratic-best-reply}) to obtain  
\[
\bar{x}_{i}(\theta )=\gamma \theta +\beta \left[ \lambda_{j} \bar{x}_{j}(\theta)+(1-\lambda_{j} )x^{\ast }_{j}\right]. 
\]%
Taking expectations of both sides of this equation w.r.t $\mu$ for both $i=1,2$, we obtain $x^{\ast
}_{i}=\gamma \theta ^{\ast }/(1-\beta )$. It follows that
\begin{equation}
\bar{x}_{i}(\theta )=\frac{\gamma(1+\lambda_{j}\beta)}{1-\lambda_{i}\lambda_{j}\beta^{2}} \theta+\kappa\theta^{\ast},
\label{eq: quadratic-x}
\end{equation}%
where $\kappa$ is a coefficient that involves $\beta$, $\lambda_{i}$, and $\lambda_{j}$.

Using the normalization $Var(\theta )=1$, we obtain

\begin{equation}
Var(\bar{x}_{i})=\left( \frac{\gamma(1+\lambda_{j}\beta) }{1-\lambda_{i}\lambda_{j}\beta^{2}}\right)^{2} \label{eq: var x}
\end{equation}

By Proposition \ref{prop: fine-vs-coarse-partition}, $\lambda_{i} =\bbP \left\{ Var(\varepsilon )\leq Var(\bar{x}_{i})\right\} $. Plugging (\ref{eq: var x}) into this equation gives us a pair of equations for $\lambda_{1}$ and $\lambda_{2}$.

When we restrict attention to symmetric CVE, the characterization is reduced to a single equation for $\lambda$:\footnote{When $\theta $ is uniformly distributed, the R.H.S. of this equation can be
conveniently written in terms of the $cdf$ of a $\chi^{2}$ distribution with $\left\vert \Theta \right\vert $ degrees of freedom.}
\begin{equation}
\lambda_{i} =\bbP \left\{ Var(\varepsilon )\leq \left( \frac{\gamma }{1-\beta
\lambda }\right) ^{2}\right\}  \label{eq: quadratic-lambda}
\end{equation}

Let us now apply this characterization to two special cases of interest.\bigskip

\noindent \textit{Strategic substitutability}

\noindent Suppose $\beta <0$. In this case, the game exhibits \textit{strategic substitutability}. The standard Cournot competition with linear demand (and allowing prices to get negative values) is a familiar instance. In this case, equation (\ref{eq: quadratic-lambda}) has a \textit{unique interior solution}. The reason is that the R.H.S. is continuously decreasing in $\lambda $ when $\beta <0$; the R.H.S. is above the L.H.S. at $\lambda =0$; and the R.H.S. is below the L.H.S. at $\lambda =1$.

Moreover, when the variance $\sigma^{2}$ of the training sample noise goes up, the R.H.S. of \eqref{eq: quadratic-lambda} decreases, and therefore so does the solution $\lambda $. This means that the coarse partition is more likely to be adopted in symmetric CVE when the training sample is noisier. As a result, while the ex-ante expected equilibrium action $x^{\ast }$ is invariant to $\sigma ^{2}$, its \textit{volatility} $Var(\bar{x})$ increases with $\sigma^{2}$ when $\beta <0$, because $Var(\bar{x})$ is decreasing in $\lambda $ --- which itself is decreasing in $\sigma ^{2}$.

Thus, when the game exhibits strategic substitutability, symmetric-equilibrium behavior is more volatile relative to Nash equilibrium. In the Cournot-competition case, this means that equilibrium prices are more volatile than in Nash equilibrium. The intuition is that when a player's belief is based on the coarse partition, she underestimates the opponent's reaction to the fluctuations in $\theta $; strategic substituability then implies that the player's own reaction to $\theta $ is amplified. This effect increases $\sigma ^{2}$, through its effect on the equilibrium probablity of adopting the coarse partition.

Unlike Nash equilibrium, for some parameter values there exist \textit{asymmetric} CVE. For example, for $\beta =-0.85$ and $\gamma =1.5\sigma$, (\ref{eq: var x}) has the asymmetric solutions $(0.473,0.938)$ and $(0.938,0.473)$ for $(\lambda_1, \lambda_2)$, in addition to the unique symmetric solution $\lambda=0.794$. The intuition behind the possibility of asymmetric CVE is as follows. Suppose  player $1$'s strategy has high volatility while player $2$'s strategy has low volatility. Then, player $2$ is likely to adopt the fine partition while player $1$ is likely to adopt the coarse partition. As a result, player $2$'s mean belief is relatively responsive to the state while player $1$'s belief is relatively unresponsive. Strategic substituability then implies that player $2$'s best-reply is relatively unresponsive to the state while player $1$'s best-reply is relatively responsive. This is consistent with the initial guess regarding the volatility of the two players' behavior.

\bigskip

\noindent \textit{Team effort}

\noindent We now examine an example of a game that exhibits strategic complementarities. Consider a team of two agents whose effort choices affect
a project's outcome. If the project is successful, each agent receives a payoff of $1$, whereas an unsuccessful project yields a payoff of $0$. Agents choose effort $a_{i}\in \mathbb{R}$ simultaneously after commonly observing a uniformly distributed state of Nature $\theta \in \{\theta _{1},\theta_{2}\}$. Let
\[
u_{i}(a_{i},a_{j},\theta )=\delta \theta (a_{i}+a_{j})+(1-\delta )a_{i}a_{j}-\frac{1}{2}a_{i}^{2}.
\]
The first two terms represent the probability that the project is successful, where $\delta \in (0,1)$ measures the relative importance of complementarities in effort. The third term represents the cost of effort. The unique Nash equilibrium in state $\theta $ is $a_{1}=a_{2}=\theta $, yielding a common ex-ante payoff of $\frac{1}{2}(\frac{1}{2}+\delta )(\theta_{1}^{2}+\theta _{2}^{2})$.

We now turn to symmetric CVE. Let $X=\mathbb{R}$, and identify $x_{i}$ with $a_{j}$ --- i.e., each agent's ML agent aims to predict the opponent's
action. Note that the expression for $u_{i}$ does not literally fit into  \eqref{eq: quadratic-u}, because it contains an additional term that is linear in $a_{j}$. However, this term matters for welfare analysis and interpretation only; it is behaviorally irrelevant, and therefore we can apply our previous characterization.

Because there are only two states, $Var(\varepsilon )\leq Var(\bar{x})$ if and only if $(\Delta \varepsilon )^{2}<(\Delta \bar{x}_{i})^{2}$, where $\Delta \bar{x}=\bar{x}(\theta _{2})-\bar{x}(\theta _{1})$ and $\Delta\varepsilon =\varepsilon _{\theta _{1}}-\varepsilon _{\theta _{2}}$.
Equation \eqref{eq: quadratic-x} becomes
\[
\bar{x}(\theta )=\delta \theta +(1-\delta )\left[ \lambda \bar{x}(\theta
)+(1-\lambda )\frac{\bar{x}(\theta _{1})+\bar{x}(\theta _{2}))}{2}\right] 
\]%
such that
\[
\Delta \bar{x}=\frac{\delta \Delta (\theta )}{1-\lambda (1-\delta )} 
\]%
Consequently, equation \eqref{eq: quadratic-lambda} becomes%
\[
\lambda =\bbP \left( \left\vert \Delta \varepsilon \right\vert <\frac{\delta
\Delta \theta }{1-\lambda (1-\delta )}\right) 
\]%
Using the fact that $\Delta \varepsilon \sim N(0,2\sigma _{\varepsilon}^{2}) $, we can rewrite this equation as follows:
\begin{equation}
\lambda =2\Phi \left( \frac{\delta \Delta \theta }{\sqrt{2}\sigma_{\varepsilon }[1-\lambda (1-\delta )]}\right) -1
\label{eq: team-equation-normal}
\end{equation}
where $\Phi (\cdot )$ is the \textit{cdf} of the standard normal distribution.

A solution to \eqref{eq: team-equation-normal} always exists. However, in general, the solution is not unique; see Figure \ref{fig: team-effort}. The reason such multiplicity can arise is that the R.H.S. of (\ref{eq: team-equation-normal}) is increasing in $\lambda $. The intuition behind this feature is that if ML agents adopt the fine partition more frequently, team members' equilibrium beliefs tend to become less coarse and more volatile. The team effort game's strategic complementarities then imply that agents' best-reply also becomes more volatile, which implies that the fine partition is more likely to prevail.

Let us explore \eqref{eq: team-equation-normal} from additional angles. First, fix $\delta $. When $\sigma _{\varepsilon }^{2}\rightarrow \infty $, $\lambda \rightarrow 0$, such that $\Delta \bar{x}\rightarrow \delta \Delta\theta $. Conversely, when $\sigma _{\varepsilon }^{2}\rightarrow 0$, $\lambda \rightarrow 1$, such that $\Delta \bar{x}\rightarrow \Delta \theta $. Second, fix $\sigma _{\varepsilon }>0$. As $\delta \rightarrow 0$, $\lambda \rightarrow 0$ --- i.e., agents' ML agents almost always adopt the coarse partition. And since $\delta $ is small, agents' equilibrium behavior is arbitrarily rigid --- i.e., $x(\theta )\approx (\theta _{1}+\theta_{2})/2 $ for every $\theta $ --- and when $\sigma _{\varepsilon }$ is small, their equilibrium payoff in each state is approximately $(\theta _{1}+\theta _{2})^{2}/8$. Compare this with the Nash equilibrium payoff $(\theta _{1}+\theta _{2})^{2}/4$ in the same limit. In this sense, CVE induces a substantial deviation from Nash equilibrium for any training sample noise variance.\footnote{The order of limits clearly matters here. For any given $\delta $, if we send $\sigma _{\varepsilon }^{2}$ to zero, agents almost always select the fine partition and their equilibrium behavior converges to the Nash benchmark.} In general, agents' expected payoff in CVE is below the Nash equilibrium level.

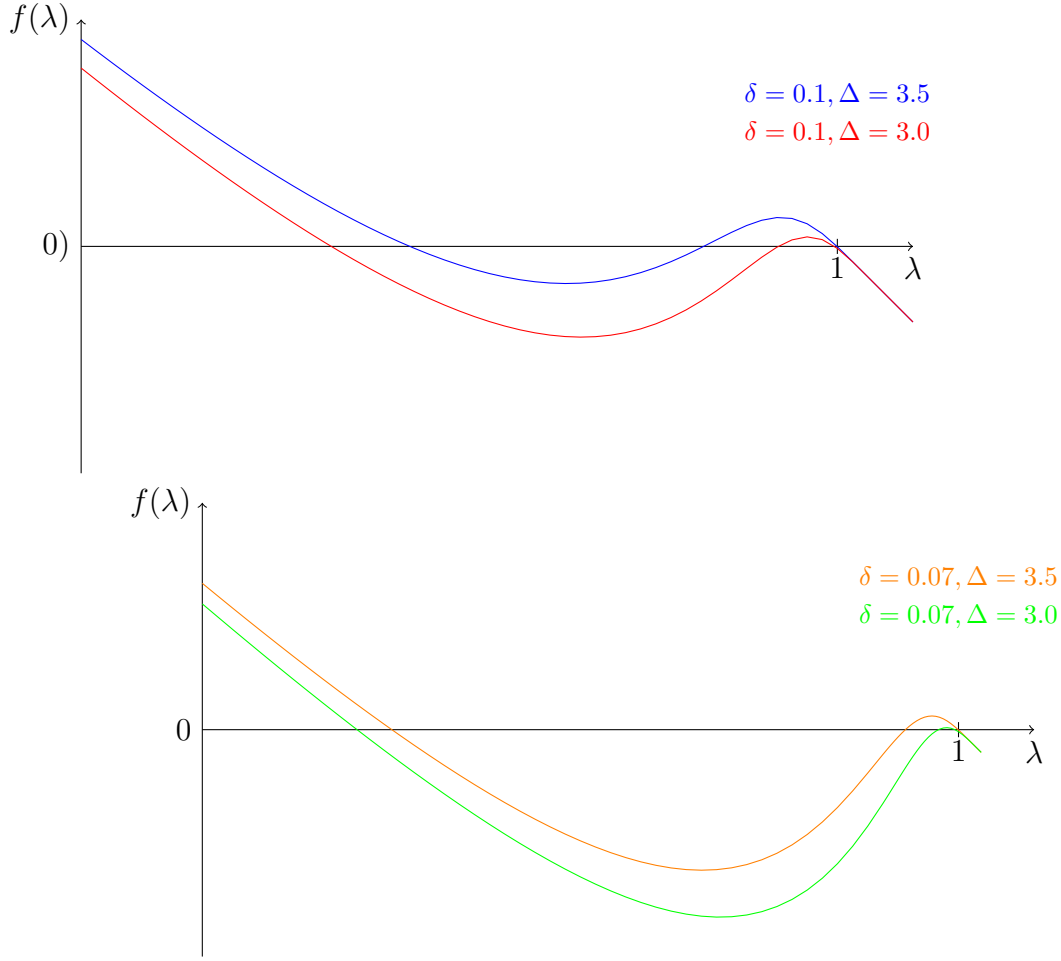
\begin{figure}
    \begin{tikzpicture}[scale=10]
        \draw[->] (0, 0) -- (1.1, 0);
        \draw[->] (0, -0.3) -- (0, 0.3);
        \node[below] at (1.1, 0) {$\lambda$};
        \node[left] at (0, 0.3) {$f(\lambda)$};
        \node[left] at (0, 0) {$0)$};
        \draw (1, -0.01)--(1, 0.01);
        \node[below] at (1, 0) {$1$};

     \draw[blue]  (0., 0.273661)--(0.02, 0.258471)--(0.04, 0.243448)--(0.06, 0.228602)--(0.08, 0.213942)--(0.1, 0.199478)--(0.12, 0.18522)--(0.14, 0.171181)--(0.16, 0.157371)--(0.18, 0.143805)--(0.2, 0.130496)--(0.22, 0.11746)--(0.24, 0.104712)--(0.26, 0.0922708)--(0.28, 0.0801542)--(0.3, 0.0683829)--(0.32, 0.0569788)--(0.34, 0.0459654)--(0.36, 0.0353683)--(0.38, 0.0252151)--(0.4, 0.0155354)--(0.42, 0.00636147)--(0.44, 
-0.00227195)--(0.46, -0.0103274)--(0.48, -0.0177644)--(0.5, -0.0245394)--(0.52, -0.0306058)--(0.54, -0.0359133)--(0.56, -0.0404088)--(0.58, -0.0440358)--(0.6, -0.046735)--(0.62, -0.0484451)--(0.64, -0.0491038)--(0.66, -0.0486496)--(0.68, -0.047024)--(0.7, -0.0441762)--(0.72, -0.040068)--(0.74, -0.0346828)--(0.76, -0.0280368)--(0.78, -0.0201962)--(0.8, -0.0112995)--(0.82, -0.0015892)--(0.84, 0.00854933)--(0.86, 0.0185395)--(0.88, 0.0275653)--(0.9, 0.0345402)--(0.92, 0.0381374)--(0.94, 0.0369574)--(0.96, 0.0299333)--(0.98, 0.016984)--(1., -0.000465258)--(1.02, -0.0200197)--(1.04, -0.04)--(1.06, -0.06)--(1.08, -0.08)--(1.1, -0.1);  
    \node[blue] at (1, 0.2) {\footnotesize{$\delta=0.1, \Delta=3.5$}};
    \draw[red] (0., 0.235823)--(0.02, 0.220014)--(0.04, 0.204354)--(0.06, 0.188851)--(0.08, 0.173514)--(0.1, 0.158351)--(0.12, 0.143373)--(0.14, 0.128589)--(0.16, 0.114012)--(0.18, 0.0996531)--(0.2, 0.0855257)--(0.22, 0.0716439)--(0.24, 0.058023)--(0.26, 0.0446794)--(0.28, 0.0316311)--(0.3, 0.0188973)--(0.32, 0.00649925)--(0.34, -0.00554035)--(0.36, -0.0171966)--(0.38, -0.0284424)--(0.4, -0.0392483)--(0.42, -0.0495821)--(0.44, -0.0594086)--(0.46, -0.0686892)--(0.48, -0.077382)--(0.5, -0.0854409)--(0.52, -0.0928155)--(0.54, -0.0994507)--(0.56, -0.105286)--(0.58, -0.110256)--(0.6, -0.114289)--(0.62, -0.117307)--(0.64, -0.119227)--(0.66, -0.119958)--(0.68, -0.119406)--(0.7, -0.117474)--(0.72, -0.114063)--(0.74, -0.109077)--(0.76, -0.102434)--(0.78, -0.0940735)--(0.8, -0.0839768)--(0.82, -0.0721933)--(0.84, -0.0588813)--(0.86, -0.0443653)--(0.88, -0.0292156)--(0.9, -0.0143481)--(0.92, -0.00112665)--(0.94, 0.00859125)--(0.96, 0.0126078)--(0.98, 0.00898974)--(1., -0.0026998)--(1.02, -0.0202537)--(1.04, -0.0400028)--(1.06, -0.06)--(1.08, -0.08)--(1.1, -0.1);
    \node[red] at (1, 0.15) {\footnotesize{$\delta=0.1, \Delta=3.0$}};
\end{tikzpicture}

\begin{tikzpicture}[scale=10]
        \draw[->] (0, 0) -- (1.1, 0);
        \draw[->] (0, -0.3) -- (0, 0.3);
        \node[below] at (1.1, 0) {$\lambda$};
        \node[left] at (0, 0.3) {$f(\lambda)$};
        \node[left] at (0, 0) {$0$};
        \draw (1, -0.01)--(1, 0.01);
        \node[below] at (1, 0) {$1$};

    \draw[orange] (0., 0.193544)--(0.02, 0.177137)--(0.04, 0.160865)--(0.06,  0.144734)--(0.08, 0.128754)--(0.1, 0.112933)--(0.12,  0.0972804)--(0.14, 0.0818065)--(0.16, 0.0665223)--(0.18, 0.05144)--(0.2, 0.0365724)--(0.22, 0.0219337)--(0.24,  0.00753937)--(0.26, -0.00659396)--(0.28, -0.020448)--(0.3, -0.0340028)--(0.32, -0.0472364)--(0.34, -0.0601249)--(0.36, -0.0726419)--(0.38, -0.0847584)--(0.4, -0.0964424)--(0.42, -0.107658)--(0.44, -0.118367)--(0.46, -0.128526)--(0.48, -0.138086)--(0.5, -0.146993)--(0.52, -0.155187)--(0.54, -0.162602)--(0.56, -0.169163)--(0.58, -0.174785)--(0.6, -0.179374)--(0.62, -0.182826)--(0.64, -0.185021)--(0.66, -0.185829)--(0.68, -0.185101)--(0.7, -0.182676)--(0.72, -0.178375)--(0.74, -0.172008)--(0.76, -0.163376)--(0.78, -0.152282)--(0.8, -0.138551)--(0.82, -0.122066)--(0.84, -0.102823)--(0.86, -0.0810367)--(0.88, -0.0572988)--(0.9, -0.0328216)--(0.905, -0.0268137)--(0.91, -0.0209328)--(0.915,-0.0152292)--(0.92, -0.00975795)--(0.925, -0.00457959)--(0.93, 0.00024087)--(0.935, 0.00463446)--(0.94, 0.00852946)--(0.945, 0.0118532)--(0.95, 0.0145346)--(0.955, 0.0165072)--(0.96, 0.0177134)--(0.965, 0.0181093)--(0.97, 0.0176699)--(0.975, 0.0163946)--(0.98, 0.0143119)--(0.985, 0.0114817)--(0.99, 0.00799529)--(0.995, 0.0039692)--(1.0, -0.000465258)--(1.005, -0.00517751)--(1.01, -0.0100543)--(1.015, -0.0150124)--(1.02, -0.0200019)--(1.025, -0.0250002)--(1.03, -0.03);
    \node[orange] at (1, 0.2) {\footnotesize{$\delta=0.07, \Delta=3.5$}};

    \draw[green] (0., 0.166332)--(0.02, 0.149437)--(0.04, 0.13266)--(0.06, 0.116006)--(0.08, 0.0994828)--(0.1, 0.0830989)--(0.12, 0.0668622)--(0.14, 0.0507819)--(0.16, 0.0348677)--(0.18, 0.0191305)--(0.2, 0.00358172)--(0.22, -0.011766)--(0.24, -0.0268988)--(0.26, -0.0418018)--(0.28, -0.0564585)--(0.3, -0.0708508)--(0.32, -0.0849592)--(0.34, -0.0987617)--(0.36, -0.112234)--(0.38, -0.125351)--(0.4, -0.138082)--(0.42, -0.150394)--(0.44, -0.162252)--(0.46, -0.173616)--(0.48, -0.184439)--(0.5, -0.194672)--(0.52, -0.204257)--(0.54, -0.21313)--(0.56, -0.221219)--(0.58, -0.228442)--(0.6, -0.234706)--(0.62, -0.239905)--(0.64, -0.243918)--(0.66, -0.246607)--(0.68, -0.247815)--(0.7, -0.247361)--(0.72, -0.245042)--(0.74, -0.240624)--(0.76, -0.233846)--(0.78, -0.224421)--(0.8, -0.212038)--(0.82, -0.196382)--(0.84, -0.177166)--(0.86, -0.154201)--(0.88, -0.127523)--(0.9, -0.0976264)--(0.905, -0.0897813)--(0.91, -0.0818455)--(0.915, -0.0738579)--(0.92, -0.0658649)--(0.925, -0.0579207)--(0.93, -0.050088)--(0.935, -0.0424389)--(0.94, -0.0350547)--(0.945, -0.0280267)--(0.95, -0.021455)--(0.955, -0.0154476)--(0.96, -0.010118)--(0.965, -0.00558141)--(0.97, -0.00194917)--(0.975, 0.000678406)--(0.98, 0.00222169)--(0.985, 0.00263277)--(0.99, 0.00190715)--(0.995, 0.0000937872)--(1., -0.0026998)--(1.005, -0.00631143)--(1.01, -0.0105409)--(1.015, -0.0151792)--(1.02, -0.020044)--(1.025, -0.0250071)--(1.03, -0.0300006);
     \node[green] at (1, 0.15) {\footnotesize{$\delta=0.07, \Delta=3.0$}};
    \end{tikzpicture}
        \centering
    \caption{\small{This figure plots $f(\lambda)=2 \Phi \left(\frac{\delta (\theta_2-\theta_1) }{\sqrt{2} \sigma_{\varepsilon}
[1-\lambda (1-\delta )]} \right) -\lambda - 1$ for different values of $\delta$ and $\Delta=\frac{ \theta_2-\theta_1 }{\sqrt{2} \sigma_{\varepsilon}}$. A root of $f(\lambda)$ corresponds to a solution to Equation \eqref{eq: team-equation-normal}. Note that there are three roots in $(0, 1)$ for each parameter specification.}}
    \label{fig: team-effort}
\end{figure}

\subsection{Unanimity Voting}

The following is a pedagogically convenient specification of the familiar jury model (\citet{FeddersenPesendorfer1998}). There are two players who vote on a reform. Voting is unanimous so that the reform is adopted if and only if both players vote for it. The voters have aligned preferences and receive a state-dependent payoff from the reform: this payoff is $1$ in state $\tht_1$ and $-c<0$ in state $\tht_2$. The payoff is $0$ in both states if the reform is not adopted. The two states are equally likely.

Before voting, players receive a signal $t_i\in \{g, b\}$ that is informative of the state. Signals are independent across voters conditional on the state. Assume $\bbP[t_i=g |\tht =\tht_1]=\bbP[t_i=b |\tht =\tht_2]= q \in (1/2, 1)$. 

Denote voting for the reform by $Y$ and voting against by $N$.
The sufficient statistic of player $i$ is $f(\tht_1, a_j)= \ind\{a_j=Y\}$ and $f(\tht_2, a_j)= -c  \ind\{a_j=Y\}$.
Player $i$'s payoff function is then
\[u_i(t_i, Y, x)=x \quad u_i(t_i, N, x)=0. \]
Thus, the sufficient statistic is the payoff from voting for reform.

Denote by $\sigma_i(t_i)$ the probability that a player votes $Y$ after signal $t_i$. Then, the expected sufficient statistic is
\begin{align*}
    \overline{x}(g)=(q^2-c(1-q)^2) \sigma_j(g) + q(1-q)(1-c)\sigma_j(b)\\
    \overline{x}(b)=(q^2-c(1-q)^2) \sigma_j(b) + q(1-q)(1-c)\sigma_j(g)
\end{align*}

Assume that $\frac{1-q}{q}<c<1$. Under this parametric restriction, there are two symmetric Nash equilibria: $\sigma(g)=\sigma(b)=0$ and $\sigma(g)=1, \sigma(b)=q(1-q)(1-c)/(cq^2-(1-q)^2)$.

Let us explore symmetric CVE, under the assumption that the training sample noise is distributed uniformly on $\{-d, d\}$.

First, we argue that $\sigma(g)=\sigma(b)=0$ is not a CVE for any positive level of noise. To see this, observe that a player votes in favour of the reform if her belief about the sufficient statistic is positive. For $\sigma(g)=\sigma(b)=0$, $\overline{x}(g)=\overline{x}(b)=0$. Consequently, a player votes $Y$ if the noise realization is positive. This rules out the uninformative, always-vote-against Nash equilibrium.

Second, we argue that the probability of voting $Y$ can be \textit{higher} in CVE than in the efficient Nash equilibrium. For ease of exposition, assume $q=2/3$ and $c=5/8$. This parameter selection implies that in the efficient Nash equilibrium, $\sigma(b)=1/2$. Let $d=0.22$. The following is the unique symmetric CVE: $\sigma(g)=1, \sigma(b)=8 d - 7/6$. Note that $1/2<\sigma(b)<3/4$. Given these strategies, $\bar{x}(g)-\bar{x}(b)=2 d$. Thus,  the fine partition is chosen when 
\[(\ve_g, \ve_b)\in \{(d, d), (-d, -d), (d, -d)\}, \]
and the coarse partition is chosen when $(\ve_g, \ve_b)=(d, -d)$. At the realisation $(\ve_g, \ve_b)=(-d, d)$, the fine and the coarse partition achieve the same out-of-sample prediction error. The coarse partition is chosen with probability $32 d - 20/3$.
Players' best-responses are to vote $Y$ after the $g$ signal, irrespective of the chosen partition. When the fine partition is chosen, a player votes $Y$ after the $b$ signal if and only if $\ve_b=d$; for the noise realization $\ve_b=-d$, the player votes $N$ after the $b$ signal when the fine partition is chosen.
When the coarse partition is chosen, a player votes $Y$ after the $b$ signal as well.

Thus, although one could expect that the training sample noise would make it harder for voters to coordinate on unanimous support for reform, the endogenous model selection ends up pushing them to vote for reform more frequently than in the efficient Nash equilibrium. Indeed, endogenous model selection plays a key role in this effect. If players were forced to use the fine partition, the unique CVE would be $\sigma(g)=1, \sigma(b)=1/2$. In contrast, in the CVE we constructed, players vote $Y$ when the coarse partition is selected.

\subsection{Speculative Bets}

Consider a two player speculative betting game. There are two commonly known states of the world, $\tht_1=1$ and $\tht_2=-1$. Both states are ex-ante equally likely. 
Player $i$ can either accept the bet $a_i=A$ or decline $a_i=D$. The bet is executed if and only if both players accept. 
The payout to player $1$ if the bet is accepted equals the state $\tht$; player 2 gets the negative.

The sufficient statistic is the probability $x_i$ with which player $j$ accepts the bet, i.e., $x_i=\ind\{a_j=A\}$. That is, $x_i$ is the position player $j$ takes in the bet. We do not restrict $x_i\in [0, 1]$ so as to allow players to adopt short and long positions.

The payoffs then are $u_1(\tht, A, x_1)=\tht x_1$ and $u_1(\tht, D, x_1)=0$ and $u_2(\tht, A, x_2)=- \tht x_2$ and $u_2(\tht, D, x_2)=0$.

Player 1 has access only to the fine partition $\Pi^f=\{\{\tht_1\},\{\tht_2\}\}$.
Player 2 has access to the fine partition and the coarse partition $\Pi^c=\{\{\tht_1,\tht_2\}\}$.
Assume the training sample noise $\ve_\tht$ is distributed uniformly on $\{-d, d\}$ for a $d>0$.

When the noise $d$ is between $1/2$ and $3/4$, the unique CVE has
\begin{align*}
    \sigma_1(A|\tht_1)=\frac{1}{2} & \sigma_1(A|\tht_2)=0\\
     \sigma_1(A|\tht_2)=\frac{1}{4} & \sigma_1(A|\tht_2)=\frac{3}{4}.
\end{align*}

The corresponding payoffs for player 1 are $u_1(\tht_1)=1/8$ and $u_1(\tht_2)=0$, and player 2 gets the negative. Player 1 thus receives a higher payoff than player 2 and a higher payoff than in Nash equilibrium.

We now argue that the proposed strategies constitute a CVE.\footnote{The proof of uniqueness is straightforward but tedious.}
Since Player 1 always uses the fine partition, she plays $A$ in state $\tht_1$ iff
$\sigma_2(A\mid\tht_1)+\ve_{\tht_1}\geq 0$,
and plays $A$ in state $\tht_2$ iff
$\sigma_2(A\mid\tht_2)+\ve_{\tht_2}\leq 0$. Since $1/2<d<3/4$, this implies  $\sigma_1(A|\tht_1)=1/2$ and $\sigma_1(A|\tht_2)=0$.

With two equally likely states, player 2's fine partition is selected over the coarse partition iff
$|\ve_{\tht_1}-\ve_{\tht_2}|<\left|\sigma_1(A\mid\tht_1)-\sigma_1(A\mid\tht_2)\right|$.
Since  $d>\frac12$, then the coarse partition is selected whenever the two noise realisations have opposite signs. Player 2 selects the fine partition exactly when $\ve_{\tht_1}=\ve_{\tht_2}$, and selects the coarse partition when $\ve_{\tht_1}\neq\ve_{\tht_2}$.
When the noise realisation is $(\ve_{\tht_1},\ve_{\tht_2})=(d,d)$, the fine partition is selected, and player 2  plays $D$ in state $\tht_1$ and $A$ in state $\tht_2$.
When the noise realisation is  $(\ve_{\tht_1},\ve_{\tht_2})=(-d,-d)$, the fine partition is selected, and player 2 plays $A$ in state $\tht_1$ and $D$ in state $\tht_2$.
For the realisation $(\ve_{\tht_1},\ve_{\tht_2})=(d,-d)$, the coarse partition is selected. The pooled estimate is
$\frac12\left(\frac12+d-d\right)=\frac14>0$,
hence Player 2 plays $D$ in state $\tht_1$ and $A$ in state $\tht_2$.
The same holds for the realisation $(\ve_{\tht_1},\ve_{\tht_2})=(-d, d)$.
Thus Player 2 plays $A$ in state $\tht_1$ only in the realisation $(-d,-d)$, which has probability $1/4$.

When the noise satisfies $d>\frac34$, the unique CVE has
\begin{align*}
    \sigma_1(A|\tht_1)=\frac{1}{2} & \sigma_1(A|\tht_2)=\frac{1}{2}\\
     \sigma_1(A|\tht_2)=\frac{1}{4} & \sigma_1(A|\tht_2)=\frac{3}{4}.
\end{align*}

The corresponding payoffs for player 1 are $u_1(\tht_1)=1/8$ and $u_1(\tht_2)=-3/8$, and player 2 gets the negative. In expectation over the state, player 1 thus receives a strictly lower payoff than player 2 and a lower payoff than in Nash equilibrium.

Player 1's estimate in state $\tht_1$ is
$\frac14+\ve_{\tht_1}$, hence Player 1 plays $A$ in state $\tht_1$ exactly when $\ve_{\tht_1}=d$. Therefore $\sigma_1(A\mid\tht_1)=\frac12$. In state $\tht_2$, Player 1's estimate is
$\frac34+\ve_{\tht_2}$, so player 1 plays $A$ in state $\tht_2$ exactly when $\ve_{\tht_2}=-d$.

Player 2's statistic is constant across states. Thus the fine partition is never strictly better than the coarse partition. If the noise realisations have opposite signs, the coarse partition is strictly better. If they have the same sign, the two partitions generate the same estimate, so either partition is optimal.

Again consider the four equally likely noise realisations.
For the realisation $(\ve_{\tht_1},\ve_{\tht_2})=(d,d)$, player 2 plays $D$ in state $\tht_1$ and $A$ in state $\tht_2$.
For the realisation  $(\ve_{\tht_1},\ve_{\tht_2})=(-d,-d)$, player 2 plays $A$ in state $\tht_1$ and $D$ in state $\tht_2$.

For the realisations   $(\ve_{\tht_1},\ve_{\tht_2})=(d,-d)$ and $(\ve_{\tht_1},\ve_{\tht_2})=(-d,d)$, the coarse partition is selected. The average sufficient statistic is
\[\frac12\left(\frac12+d+\frac12-d\right)=\frac12>0.\]
Thus Player 2 plays $D$ in state $\tht_1$ and $A$ in state $\tht_2$.
Thus Player 2 plays $A$ in state $\tht_1$ only in the realisation $(-d,-d)$, which has probability $1/4$.

The example shows that access only to the correct fine partition need not benefit a player, even in a zero-sum setting. At high noise levels, player 1 obtains a lower payoff than player 2. Player 2’s ability to choose the coarse partition therefore raises her payoff in this case. The intuition is as follows. Player $i$ accepts the bet in the state that is unfavourable to her, $\theta_j$, only if she believes that the other player is taking a short position. When beliefs are very noisy, this occurs with positive probability and leads to a negative payoff. By contrast, when player 2 chooses the coarse partition---which happens when the noise realizations have opposite signs in the two states---the coarse partition averages out the noise. Her belief is then correct, and she chooses the weakly dominant action in each state. The example shows that this effect can be strong enough to generate non-zero payoffs.

\subsection{An Illusion of Control}

We use the final example in this section to explore a natural extension of the framework, and a notable effect that this extension generates.

Consider a binary-action, two-person game with a commonly observed binary state. Let $a_{i}\in \{0,1\}$. Identify $x_{i}$ with the probability that $a_{j}=1$. Let $u_{i}(\theta ,a_{i},x_{i})=2\theta x_{i}-a_{i}$, where $\theta \sim U\{0,1\}$. Under rationality, $a_{i}=0$ is strictly dominant.

We extend the main model by allowing players to form beliefs that are based on partitions of $\Theta \times A$ --- whereas in the main model, only
partitions of $\Theta $ are allowed. Specifically, we allow for three possible partitions of all four pairs $(\theta ,a)$: The maximally fine partition, and two coarse partitions: $\{\{(0,0),(0,1)\},\{(1,0),(1,1)\}\}$ and $\{\{(0,0),(1,0)\},\{(0,1),(1,1)\}\}$. The former coarse partition is correct, because  player $j$'s action is purely a function of $\theta $. The latter coarse partition is wrong, because it falsely attributes $a_{j}$ to $a_{i}$. By Proposition \ref{prop: redundant-partition}, we can ignore the maximally fine partition because it is a refinement of the correct coarse partition. Therefore, in symmetric CVE, players' ML agents will only mix between the two coarse partition.

Since players' ML agents may attribute outcomes to both states and actions, player $i$'s training sample noise is associated with pairs $(\theta ,a_{i})$. Suppose $\varepsilon _{\theta a}\sim U\{-\frac{1}{2},\frac{1}{2}\}$.

In CVE, $a_{i}=0$ with certainty when $\theta =0$, independently of the partition player $i$'s ML agent adopts. Let $\alpha $ denote the probability that $a_{i}=1$ at $\theta =1$ in CVE. When player $i$'s ML agent adopts the correct coarse partition, she chooses the rational action $a_{i}=0$, because she does not think that $a_{i}$ affects $a_{j}$. It follows that the player can only play $a_{i}=1$ with positive probability when $\theta =1$, provided that her ML agent adopts the wrong coarse partition.

Let us derive the equilibrium partition selection rule. The expected out-of-sample prediction error induced by the correct coarse partition is%
\[
L_{\theta }=\frac{1}{2}\left( \frac{\varepsilon _{11}+\varepsilon _{10}}{2}\right) ^{2}+\frac{1}{2}\left( \frac{\varepsilon _{01}+\varepsilon _{00}}{2}\right) ^{2} 
\]
To derive the expected out-of-sample prediction error induced by the wrong coarse partition, note first that $a_{i}=1$  with probability $\frac{1}{2}\alpha $ and $a_i=0$ with probability $1-\frac{1}{2}\alpha $. Let us now calculate the joint equilibrium distribution over $\theta$ and $x_{i}$ conditional on any given $a_{i}$:
\[
\bbP (\theta ,x_{i}\mid a_{i}=1)=\left\{
\begin{array}{ccc}
\alpha & for & (\theta ,x_{i})=(1,1) \\ 
1-\alpha & for & (\theta ,x_{i})=(1,0)%
\end{array}%
\right. 
\]%
and
\[
\bbP (\theta ,x_{i}\mid a_{i}=0)=\left\{ 
\begin{array}{ccc}
\frac{\frac{1}{2}(1-\alpha )}{1-\frac{1}{2}\alpha }\alpha & for & (\theta
,x_{i})=(1,1) \\ 
\frac{\frac{1}{2}(1-\alpha )}{1-\frac{1}{2}\alpha }(1-\alpha ) & for & 
(\theta ,x_{i})=(1,0) \\ 
\frac{\frac{1}{2}}{1-\frac{1}{2}\alpha } & for & (\theta ,x_{i})=(0,0)%
\end{array}%
\right. 
\]
The expected out-of-sample prediction error of the wrong coarse partition is thus
\[
L_{a}=\frac{\alpha }{2}\varepsilon _{11}^{2}+\frac{1-\alpha }{2}\left(\varepsilon _{10}-\frac{\alpha (1-\alpha )}{2-\alpha }\right) ^{2}+\frac{1}{2}\left( \varepsilon _{00}-\frac{\alpha (1-\alpha )}{2-\alpha }\right) ^{2} 
\]

We can now calculate $\alpha $: It is the probability that the following inequalities hold jointly: $L_{a}<L_{\theta }$ and $2(\alpha +\varepsilon
_{11})-1>0$. The logic behind the latter inequality is that conditional on $a_{i}=1$, $\theta =1$ and $x_{i}=\alpha $; and so when player $i$ adopts the wrong coarse partition, her estimate of $x_{i}$ at $a=1$ is $\alpha+\varepsilon _{11}$.

It can be checked that the configuration $\varepsilon _{\theta a}=\frac{1}{2}$ for all $(\tht  a)$ (which occurs with probability $\frac{1}{16}$) together with $\alpha=1/16$ is the only one that satisfies both inequalities. Therefore, this specification constitutes a CVE. We have thus found a CVE in which players choose the dominated action $a=1$ with positive probability. This is an \textquotedblleft illusion of control\textquotedblright\ effect: Although $a_{i}$ has no objective causal effect on $a_{i}$, players sometimes choose the costly action $a=1$ because the model their ML agent adopts sometimes attributes a successful outcome to this action.

\section{Related Literature}
\label{sec: lit}


\noindent Our main contribution is to the literature on how ML plays out in strategic interactions, where a key feature is that the training data is endogenous. While we incorporate ML into belief equilibrium formation and assume that decisions are made by humans in response to ML-generated beliefs, most of the literature in economics focuses on reinforcement learning. \citet{CalvanoEtAl2020}, \citet{HansenPai2021}, \citet{AskerEtAl2024}, among many others, numerically study reinforcement learning in oligopoly games.   \citet{BanchioMantegazza2024},  \citet{Dolgopolov2024}, \citet{Possnig2024}, and \citet{CarteaEtAl2025} analytically characterize long-run outcomes in such environments.
 \citet{Waizmann2025} studied an interaction between a long-run player who obeys Q-learning and a sequence of rational short-run players.

More recently, \citet{JehielWeber2025}, \citet{JehielMohlin2026} and \citet{Spiegler2026} focused on ML-based predictive models that serve human players in games. All three use partitions of a set of contingencies to conceptualize predictive models. Moreover, in contrast to our model, these three papers assume that the data which determines model selection is noiseless. In \citet{JehielWeber2025}, the partition size is fixed, and the equilibrium belief formation model obeys a form of k-means clustering. Their \emph{clustered distributional analogy-based expectations equilibrium} obtains as a special case of our cross-validation equilibrium for a particular choice of feasible models and noiseless training samples. 

The model of \citet{JehielMohlin2026}, applied to static Bayesian games, differs from ours in two important aspects. First, they impose an exogenous feedback constraint, which implies that certain situations cannot be distinguished. Second, their model-selection criterion relies on an exogenous notion of similarity between situations. Neither feature is present in our framework. Under perfect feedback, their categorization equilibrium coincides with Nash equilibrium.
In \citet{Spiegler2026}, the equilibrium partition trades off complexity and expected predictive accuracy: a player incurs a physical cost that increases in the number of cells of the chosen partitions. In contrast, in the present model, there is no explicit cost of complex models. Rather, simple models can sometimes prevail because they improve out-of-sample predictive accuracy, sacrificing bias for lowering variance.

At a broader level, this paper contributes to the literature on game-theoretic equilibrium concepts that incorporate non-rational belief formation; see \citet{Jehiel2005}, \citet{EysterRabin2005}, \citet{Esponda2008}, \citet{EspondaPouzo2016}, \citet{Spiegler2016}, and \citet{clyde2025proxy}. In each of these cases, equilibrium behavior involves subjective optimization against a belief that arises from fitting (in some sense) a subjective model to the equilibrium distribution. 

A key feature that distinguishes CVE from these works is that agents' subjective models in CVE are randomly and endogenously determined, following a procedure that involves taking samples from the equilibrium distribution. In this respect, CVE is related to equilibrium concepts in which agents best-reply to endogenously drawn samples; see \citet{OsborneRubinstein1998}, \citet{SalantCherry2020}, and \citet{DanenbergSpiegler2025}. These papers, however, feature no element of ML-based model selection.
\bigskip

\appendix

\section{Proofs}
\label{sec: proofs}

In the proofs, it will be convenient to let $h_i(t_i, \ve_i, \Pi_i)\in \Delta(A_i)$ denote a player $i$'s mixed strategy when her type is $t_i$, the realized training sample is $\ve_i$  and her ML agent adopts the partition $\Pi_i$. Moreover, we let $\lambda_i(\ve_i, \Pi_i)$ denote the probability with which player $i$'s ML agent adopts the partition $\Pi_i$ given $\ve_i$. Clearly, given $\sigma_{-i}$, if $\lambda_i(\Pi, {\ve}_i)$ is supported on partitions satisfying \eqref{eq: optimal-partition} and $h_i(t_i, {\ve_i}, \Pi)$ is supported on actions satisfying \eqref{eq: optimal-action}, one can find a distribution $\beta_{{\ve}_i}$ over pairs $(\Pi_i, g_i)$ that are compatible with $\sigma_{-i}$ and $\ve_i$ according to Definition \ref{def: compatible}, and vice versa.

\subsection{Proof of Proposition \ref{prop: existence}}
\label{proof: existence}

For each player $i$ and vector of realized type-dependent shocks $\ve_i=\left(\varepsilon_{t_{i}}\right)$,  let $\lambda_{i}\left(\varepsilon_{i}\right),  \in \Delta \mathcal{P}_{i}$ be a distribution over the set of feasible models. Denote $\lambda_{i}=\lambda_{i}\left(\varepsilon_{i},\right)_{\varepsilon_ {{i}}}$ and $\lambda=\left(\lambda_{i}\right)_{i}$.\\
Let $\sigma_{i}\left(\cdot \mid t_{i}\right) \in \Delta A_{i}$ be a mixed strategy for type $t_{i}$, $\sigma_{i}=\left(\sigma_{i}\left(\cdot \mid t_{i}\right)\right)_{t_{i}}$ and $\sigma=\left(\sigma_{i}\right)_{i}$.

Construct a correspondence

\begin{align*}
\Psi: &  \bigtimes_{i=1}^{N}\left(\Delta \mathcal{P}_{i}\right)^{|\mathrm{supp} \varepsilon| \times |T_i| } \times\bigtimes_{i=1}^{N}\left(\Delta A_{i}\right)^{\left|T_{i}\right|}  \rightrightarrows \bigtimes_{i=1}^{N}\left(\Delta \mathcal{P}_{i}\right)^{|\mathrm{supp} \varepsilon| \times\left|T_{i}\right| } \times\bigtimes_{i=1}^{N}\left(\Delta A_{i}\right)^{\left|T_{i}\right|} \\
& (\lambda, \sigma) \mapsto \Psi(\lambda, \sigma)
\end{align*}

The goal is to define $\Psi$ such that, if $\left(\lambda^{*}, \sigma^{*}\right) \in \Psi\left(\lambda^{*}, \sigma^{*}\right)$ is a fixed point of this correspondence, $\sigma^{*}$ is a CVE.\\


Given $\sigma$, for any i and $t_{i}$ define $p^{\sigma}(\theta, t_i, a)$ as in Equation \eqref{eq: p-sigma}
and $\bar{x}_{i}^{\sigma}\left(t_{i}\right)$ as in Equation \eqref{eq: x-bar-sigma}.
$\bar{x}_{i}^{\sigma}\left(t_{i}\right)$ is a continuous function of $\sigma$, for any fixed $t_{i}$.
For each $t_{i}, \varepsilon_{i}=\left(\varepsilon_{t_{i}}\right) \in \operatorname{supp} \varepsilon_{i}^{|T i|}$, and $T_i\supset M_{i}=\pi(t_i) \in \Pi_i \in \mathcal{P}_{i}$
define

\[
\hat{x}^{\varepsilon_{i}, \sigma}(M_{i})=\sum_{t_{i} \in M_{i}}  \mu(t_{i}| t_{i} \in M)\left(\bar{x}_{i}^{\sigma}\left(t_{i}\right)+\varepsilon_{t_{i}}\right)
\]

$\hat{x}^{\ve_i, \sigma}(M_i)$ is a continuous function of $\sigma$.

Denote by $L\left(\Pi_{i},  \varepsilon_{i}\right)(\sigma)$ the squared prediction error of a partition $\Pi_i\in \mathcal{P}_i$ given a vector of shocks $\ve_i$
\[L_{i}\left(\Pi, \varepsilon_{i}\right)(\sigma)=\mathbb{E}_{\eta}\left[\sum_{t_{i} \in T_{i}} \mu\left(t_{i}\right)\left[\bar{x}^{\varepsilon_{i}, \sigma}(\pi(t_i))-\bar{x}_{i}^{\sigma}\left(t_{i}\right)-\eta_{t_{i}}\right]^{2}\right].\]
Note that for fixed $\Pi,\ve_{i}, L\left(\Pi,  \varepsilon_{i}\right)(\sigma)$ is a continuous function of $\sigma$.

Define $\psi_{1}\left(i, \varepsilon_{i}, \right)(\sigma)$ to be

\[
\psi_{1}\left(i,  \varepsilon_{i}\right)(\sigma)=\left\{\lambda_{i}^{\prime}(\varepsilon_{i}) \in \Delta \mathcal{P}_{i}\middle| \lambda_{i}^{\prime}\left(\varepsilon_{i}\right)\left(\Pi_{i}\right)>0  \Rightarrow \Pi_{i} \in \argmin_{\tilde{\Pi_{i}} \in \mathcal{P}_{i}}  L_{i}\left(\tilde{\Pi_i}, \varepsilon_{i}\right)(\sigma)
\right\}.
\]

Note that the continuity of $L_{i}\left(\Pi,  \varepsilon_{i}\right)(\cdot)$ and the finiteness of $\mathcal{P}_{i} $ imply that $\psi_{1}\left(i,  \varepsilon_{i}\right)(\cdot)$ is an upper hemi-continuous, non-empty -, convex - and compact-valued correspondence.

Define $\Psi_{1}(\lambda, \sigma)=\left(\psi_{1}\left(i,  \varepsilon_{i}\right)(\sigma)\right)_{\varepsilon_{i},  i}.$
Clearly, $\Psi_{1}(\lambda, \sigma)$ is an upper hemi-continuous, non-empty-, convex-and compact-valued correspondence.\\

Now we define $\Psi_{2}(\lambda, \sigma) \in \bigtimes_{i}\left(\Delta A_{i}\right)^{\left|T_{i}\right|}$.
Fix $\Pi_{i} \in \mathcal{P}_{i}, \varepsilon_{i} \in(\operatorname{supp} \varepsilon)^{\left|T_{i}\right|}$ and $t_{i}$.
Define $\hat{p}^{\varepsilon_{i}}(y \mid M_i)(\sigma)$ as in Equation \eqref{eq: empirical-distribution}.



Define $H_{2}\left(i, t_i, \ve_{i}, M\right)(\sigma)$ as

\begin{align*}
\left\{
h_{i}\left(t_{i}, \ve_{i}, M\right) \in \Delta\left(A_{i}\right) \mid h_{i}\left(t_{i}, \varepsilon_{i}, M\right)(a_i)>0  \Rightarrow a_{i} \in \arg \max _{a_{i}^{\prime} \in A_{i}} \sum_{y \in X_{i}} \hat{p}^{\ve_{i}}(y \mid M)(\sigma) u_{i}\left(a^{\prime}, t_{i}, y\right)
\right\}
\end{align*}

\begin{claim}
For all fixed i, $\varepsilon_{i}, t_{i}, M_{i}, H_{2}\left(i, \varepsilon_{i}, t_{i}, M_{i}\right)(\sigma)$ is a non-empty-, convex- and compact-valued and upper-hemicontinuous correspondence.
\end{claim}
\begin{proof}
Only the claim that the correspondence is uhc needs proof. Take $\sigma^{n} \rightarrow \sigma, H\left(i, t_{i}, \varepsilon_{i}, M\right)\left(\sigma^{n}\right) \ni \gamma^{n} \rightarrow \gamma$.
Need to show that $\gamma \in H\left(i, t_{i}, \varepsilon_{i}, M\right)(\sigma)$.
Note that for each $n, \hat{p}^{\varepsilon}(y \mid M)\left(\sigma^{n}\right)$ is positive for at most $|M|$ many points $y \in X_{i}$. Assume wlog that for each $\left.n, \hat{p}^{n}(y \mid M) C_{0^{n}}\right)$ is positive for exactly $|M|$-many points y (if not, pass to subsequences such that $\hat{p}^{\varepsilon}(y \mid M)\left(\sigma^{n}\right)$ is positive for exactly $k$-many, $k \leq|M|$, points, and repeat the argument for each $k$ ). Note that there exists $|M|$ sequences $\left(x_{m}^{n}\right)$ such that $\hat{p}^{\varepsilon}\left(x_{m}^{n} \mid M\right)\left(\sigma^{n}\right)=\hat{p}^{\varepsilon}\left(x_{m}^{n+1} \mid M\right)\left(\sigma^{n+1}\right)=: p_{m}$ for all $n \in \mathbb{N}, m=1, \ldots,|M|$. Moreover, as $\sigma^{n} \rightarrow \sigma$, $x_{m}^{n} \rightarrow x_{m}$ for each $m$.

Hence, one can identify $\gamma^{n}$ with a probability distribution over 
\[\argmax _{a_i' \in A_i} \sum_{m=1}^M p_{m} u_{i}\left(a_{i}^{\prime}, t_{i}, x_{m}^{n}\right).\]

Because $u_{i}\left(a_{i}, t_{i}, x\right)$ is continuous in $x$, the limit $\gamma$ of $\gamma^{n}$ puts probability 1 on

\[
\argmax _{a^{\prime} \in A_{i}} \sum_{m=1}^{|M|} p_{m} u_{i}\left(a_{i}', t_{i}, x_{m}\right)
\]

which implies $\gamma \in H_{2}\left(i, \ve_{i}, t_{i}, M_{i}\right)(\sigma)$.
\end{proof}

Now define

\begin{align*}
\Psi_{2}(i, t_{i})(\lambda, \sigma)
&=\Bigl\{
\sigma_i(\cdot|t_i)' \in \Delta(A_{i}) \Bigm|
\exists\, h_{i}(t_{i},\varepsilon_i, M_{i}) \in H_{2}\!\left(i,\varepsilon_{i}, t_{i}, M_{i}\right)(\sigma)
\ \text{such that, for all } a_i \in A_i,\\
&\qquad \sigma_i'(a_i|t_i)
= \sum_{\varepsilon_i \in \operatorname{supp}(\varepsilon_i)\times |T_i|} v(\varepsilon_i)
  \sum_{\Pi_{i} \in \mathcal{P}_{i}} \lambda_{i}(\varepsilon_{i})(\Pi_{i})
   h_{i}\!\bigl(t_i', \varepsilon_i, \pi_{i}(t_i')\bigr)(a_i)
\Bigr\}.
\end{align*}

\begin{claim}
    The correspondence $\Psi$ admits a fixed point $(\lambda^*, \sigma^*)\in \Psi(\lambda^*, \sigma^*)$.
\end{claim}
\begin{proof}
$\Psi$ is non-empty-, compact- and convex-valued. Moreover, $\Psi_1$ is upper hemi-continuous. $\Psi_2$ is upper-hemi-continuous as  well: since $H_2(i, \ve_i, t_i, M_i)$ is upper hemi-continuous by the previous Claim, $\Psi_2(i, t_i)$ is upper hemi-continuous as well. Consequently, $\Psi$ admits a fixed point by Kakutani's Fixed Point Theorem.
\end{proof}

\begin{claim}
    Let $(\lambda^*, \sigma^*)\in \Psi(\lambda^*, \sigma^*)$ be a fixed point of the correspondence $\Psi$. Then $\sigma^*$ is a CVE.
\end{claim}
\begin{proof}
    Let $(\lambda^*, \sigma^*)\in \Psi(\lambda^*, \sigma^*)$. By  construction, for every profile $(\ve_i)$, $\lambda_i^*(\ve_i)(\Pi_i)>0$ only if $\Pi_i$ solves \eqref{eq: optimal-partition}. Moreover, for every $i$ and $t_i$,  $\sigma_i^*(\cdot|t_i)$ satisfies \eqref{eq: strategy} since it is composed of functions $h_i(t_i, \ve_i, M_i)\in H_2(i, \ve_i, t_i, M_i)(\sigma^*)$.
\end{proof}

\subsection{Lemma \ref{lemma: belief-convergence}}
\label{proof: belief-convergence}
 For a player $i$ 
    let 
    \[\Gamma(\Pi, \overline{x}, \ve, t_i)= \sum_{t_i^\prime\in \pi(t_i)} \mu(t_i^\prime| t_i^\prime \in \pi(t_i)) \left(\bar{x}_i(t_i^\prime)+\ve_{t_i^\prime}\right)\]
    be the belief of player $i$ when adopting partition $\Pi$, given the training sample noise $\ve$, the vector of sufficient statistics $\overline{x}$, and the type $t_i$.
    
\begin{lemma}
    \label{lemma: belief-convergence}
    Let $\Pi^n(\ve)$ be the chosen partition when the realization of the training sample noise is $\ve$. For each player $i$ and type $t_i$,
    \[\Gamma(\Pi^n(\ve), \overline{x}^{\sigma^n}, \ve, t_i)\xrightarrow[n\to \infty]{}  \Gamma(\Pi^{\min}, \overline{x}^{\sigma}, 0, t_i)\]
    in probability.\\
\end{lemma}
\begin{proof}
Denote by $L(\Pi, \overline{x}, \ve)$ the expected out-of-sample prediction error for partition $\Pi$ given the realized training sample $\ve$ and a vector of sufficient statistics $\overline{x}$.

First, for every $n$ and partition $\Pi\in \mathcal{P}_i$, 
\begin{align}
L(\Pi^n(\ve), \overline{x}^{\sigma^n}, \ve)\leq L(\Pi, \overline{x}^{\sigma^n}, \ve) \label{eq: bound-prediction-error}
\end{align}
almost surely.  Note that for every fixed partition $\Pi$, the prediction error $L(\Pi, \overline{x}, \ve)$ is a polynomial of degree two in $\overline{x}$ and $\ve$. Because $\overline{x}^{\sigma^n}\to \overline{x}^{\sigma}$ and $\var(\ve^n)\to 0$, the RHS of equation \eqref{eq: bound-prediction-error} converges to $L(\Pi, \overline{x}^{\sigma}, 0)$ in probability. Consequently,
\[\bbP\left[ L(\Pi^n(\ve), \overline{x}^{\sigma^n}, \ve)\leq \min_{\Pi \in \mathcal{P}_i} L(\Pi, \overline{x}^{\sigma}, 0)\right]\xrightarrow[n\to \infty]{} 1.\]
Moreover, since $\mathcal{P}_i$ is finite,
\[\bbP\left[ \Pi^n(\ve)\in  \argmin_{\Pi \in \mathcal{P}_i} L(\Pi, \overline{x}^{\sigma}, 0)\right]\xrightarrow[n\to \infty]{} 1.\]

Next, Lemma \ref{lemma: non-decreasing-SSE} and the hypothesis that $\Pi^{\textrm{finest}}$ refines every $\Pi\in \mathcal{P}_i$ imply that,
\[\Pi^{\textrm{finest}}\in \argmin_{\Pi^\prime \in \mathcal{P}_i} L(\Pi^\prime, \overline{x}^{\sigma}, 0).\]
Moreover, if
\[\Pi\in \argmin_{\Pi^\prime \in \mathcal{P}_i} L(\Pi^\prime, \overline{x}^{\sigma}, 0),\]
then for all types $t_i$
\[\Gamma(\Pi, \overline{x}^{\sigma}, 0, t_i)= \Gamma(\Pi^{\textrm{finest}}, \overline{x}^{\sigma}, 0, t_i).\]
\end{proof}

\begin{lemma}
    \label{lemma: non-decreasing-SSE}
    Let $X=\{x_i\}$ be a  finite set of numbers and $\alpha_i\in (0, 1), \sum_i\alpha_i=1$. For any $S\subset X$, denote $m(S)=\sum_{i: x_i \in S}\frac{\alpha_i x_i}{\sum_{i: x_i \in S} \alpha_i}$. Then for any $S\neq \emptyset$ 
    \[\sum_{i: x_i \in S} \frac{\alpha_i}{\sum_{i: x_i \in S} \alpha_i} (x_i-m(S))^2 + \sum_{i: x_i \in X\setminus S} \frac{\alpha_i}{\sum_{i: x_i \in X\setminus S} \alpha_i} (x_i-m(X\setminus S))^2\leq \sum_{i: x_i \in X} \frac{\alpha_i}{\sum_{i: x_i \in X} \alpha_i} (x_i-m(X))^2 \]
    with equality only if $m(S)=m(X\setminus S)$.
\end{lemma}
\begin{proof}
    The proof is standard and therefore omitted.    
\end{proof}

\subsection{Proof of Proposition \ref{prop: convergence-to-ABEE}}
\label{proof: convergence-to-ABEE}

    Let $\sigma^n$ be a CVE for each $n$. 
    Denote by $\Pi^n(\ve)$ the partition chosen given the realized training sample. Note that $\Pi^n(\ve)$ is a random variable because the training sample $\ve$ is random and because the choice of partition is random as well.

Assume toward a contradiction that for some player $i$ and type $t_i$, $\sigma_i(t_i)(a_i)=3 \beta>0$ but 
\[u_i\left(a_i, t_i, \sum_{t_i^\prime\in \Pi^{\textrm{finest}}(t_i)} \mu(t_i^\prime| t_i^\prime \in \Pi^{\textrm{finest}}(t_i)) \left(\bar{x}_i^{\sigma}(t_i^\prime)\right)\right)\]
\[\leq  u_i\left(a_i^\dagger, t_i, \sum_{t_i^\prime\in \Pi^{\textrm{finest}}(t_i)} \mu(t_i^\prime| t_i^\prime \in \Pi^{\textrm{finest}}(t_i)) \left(\bar{x}_i^{\sigma}(t_i^\prime)\right)\right) - 2 \delta\]
for some $a_i^\dagger \in A_i$ and $\delta>0$.

Let $\gamma$ be such that $u_i(a_i, t_i, x)\leq u_i(a_i^\dagger, t_i, x)-\delta$ for all $x$ within $\gamma$ of
\[\sum_{t_i^\prime\in \Pi^{\textrm{finest}}(t_i)} \mu(t_i^\prime| t_i^\prime \in \Pi^{\textrm{finest}}(t_i)) \left(\bar{x}_i^{\sigma}(t_i^\prime)\right).\]

Denote by $B_\gamma^n$ the set
\[B_\gamma^n=\left\{ (\Pi, \ve) : \left|\sum_{t_i^\prime\in \pi(t_i)} \mu(t_i^\prime| t_i^\prime \in \pi(t_i)) \left(\bar{x}_i^{\sigma^n}(t_i^\prime)+\ve_{t_i^\prime}\right) -  \sum_{t_i^\prime\in \Pi^{\textrm{finest}}(t_i)} \mu(t_i^\prime| t_i^\prime \in \Pi^{\textrm{finest}}(t_i)) \bar{x}_i^{\sigma}(t_i^\prime)\right|\leq \gamma \right\}\]

By Lemma \ref{lemma: belief-convergence},  for all $n$ large enough,
\[\sum_{(\Pi, \ve) \in B^n_\gamma} v^n(\ve)  \lambda_i^n(\Pi, \ve) \geq 1- \beta.\]

Since $\sigma_i^n\to \sigma$, it must be that for all $n$ large enough

\begin{align*}
    2 \beta & \leq \sum_{\ve} \sum_{\Pi} v^n(\ve) \lambda^n(\Pi, \ve) h_i^n(t_i, \Pi, \ve)(a_i) \leq \sum_{(\ve, \Pi)\in B_\gamma^n} v^n(\ve) \lambda^n(\Pi, \ve) h_i^n(t_i, \Pi, \ve)(a_i) + \beta\\
\end{align*}

The inequality requires that $h_i^n(t_i, \Pi, \ve)(a_i)>0$ on $B_\gamma^n$, contradicting the best-response condition, equation \eqref{eq: optimal-action} in Definition \ref{def: compatible}.

We conclude that for each player $i$ and type $t_i$, $\sigma_i(t_i)(a_i)>0$ implies
\[u_i\left(a_i, t_i, \sum_{t_i^\prime\in \Pi^{\textrm{finest}}(t_i)} \mu(t_i^\prime| t_i^\prime \in \Pi^{\textrm{finest}}(t_i)) \left(\bar{x}_i^{\sigma}(t_i^\prime)\right)\right)\geq u_i\left(a_i^\dagger, t_i, \sum_{t_i^\prime\in \Pi^{\textrm{finest}}(t_i)} \mu(t_i^\prime| t_i^\prime \in \Pi^{\textrm{finest}}(t_i)) \left(\bar{x}_i^{\sigma}(t_i^\prime)\right)\right)\]
$ \forall a_i^\dagger\in A_i$,as was to be shown.

\subsection{Proof of Proposition \ref{prop: redundant-partition}}
\label{proof: redundant-partition}

    Given the realized training sample $\ve_i$, the expected out-of-sample prediction error of partition $\Pi$ equals
    \[\bbE_{\eta
}\sum_{t_{i}\in T_{i}}\mu(t_{i})\left[\left(\sum_{t_i^\prime \in \pi(t_i)} \mu(t_i^\prime| t_i^\prime \in \pi(t_i)) \ve_{t_i^\prime}\right)-\eta _{t_{i}})\right]^{2}.\]
By Jensen's inequality, the finer partition $\Pi_i^\prime$ achieves the same out-of-sample prediction error by splitting the cell $\pi(t_i)$ into the cells $\pi^\prime_1$ and $\pi^\prime_2$ only if\footnote{The argument when splitting $\pi(t_i)$ into more than two cells is the same.}
\[ \sum_{t_i^\prime \in \pi(t_i)} \mu(t_i^\prime| t_i^\prime \in \pi(t_i)) \ve_{t_i^\prime}=\sum_{t_i^\prime \in \pi^\prime_1} \mu(t_i^\prime| t_i^\prime \in \pi^\prime_1) \ve_{t_i^\prime}=\sum_{t_i^\prime \in \pi^\prime_2} \mu(t_i^\prime| t_i^\prime \in \pi^\prime_2) \ve_{t_i^\prime}.\]

This implies that choosing the finer partition $\Pi^\prime$ induces the same beliefs as $\Pi$ for any type $t_i$.

\subsection{Proof of Proposition \ref{prop: unbiased-noise}}
\label{proof: unbiased-noise}

Fix $\varepsilon _{i}$ and $\Pi _{i}$. Consider a partition cell $M\in \Pi_{i}\,$. Denote
\[
\bbE_{\eta }\sum_{t_{i}\in M}\mu (t_{i}\mid t_{i}\in M)[\hat{x}^{\varepsilon_{i}}(M)-(\bar{x}_{i}^{\sigma }(t_{i})+\eta _{t_{i}})]^{2}=c_{M}(\varepsilon_{i}) 
\]%
Plugging the definition of $\hat{x}^{\varepsilon _{i}}$ into the L.H.S., we obtain
\[
\bbE_{\eta }\sum_{t_{i}\in M}\mu (t_{i}\mid t_{i}\in M)\left[\sum_{t_{i}^{\prime }\in M}\mu (t_{i}^{\prime }\mid t_{i}^{\prime }\in M)(\bar{x}_{i}^{\sigma }(t_{i}^{\prime })+\varepsilon _{t_{i}^{\prime }})-\bar{x}_{i}^{\sigma }(t_{i})-\eta _{t_{i}}\right] ^{2} 
\]
Expand this expression. Since $\eta _{t_{i}}$ is independently distributed with $E_{\eta }(\eta _{t_{i}})=0$, the only terms in the
elaborated expression that includes $\varepsilon _{i}$ are
\[
\left( \sum_{t_{i}^{\prime }\in M}\mu (t_{i}^{\prime }\mid t_{i}^{\prime}\in M)\varepsilon _{t_{i}^{\prime }}\right) ^{2} 
\]%
and
\[
2\sum_{t_{i}\in M}\mu (t_{i}\mid t_{i}\in M)\left( \sum_{t_{i}^{\prime }\in M}\mu (t_{i}^{\prime }\mid t_{i}^{\prime }\in M)\varepsilon _{t_{i}^{\prime }}\right) \left( \sum_{t_{i}^{\prime }\in M}\mu (t_{i}^{\prime }\mid t_{i}^{\prime }\in M)\bar{x}_{i}^{\sigma }(t_{i}^{\prime })-\bar{x}_{i}^{\sigma }(t_{i})\right) 
\]%
The second term is equal to zero. As to the first term, it remains the same if we replace $\varepsilon _{i}$ with $-\varepsilon _{i}$. It follows that $c_{M}(-\varepsilon _{i})=c_{M}(\varepsilon _{i})$ for every partition cell $M $. Therefore, if $\Pi _{i}$ satisfies \eqref{eq: optimal-partition} under $\varepsilon _{i}$, then it also does so under $-\varepsilon _{i}$. Since $\upsilon $ is symmetric around zero, it follows that the expectation of $\varepsilon _{i}$ conditional on $\Pi _{i}$ satisfying \eqref{eq: optimal-partition} is zero.

\subsection{Proof of Proposition \ref{prop: fine-vs-coarse-partition}}
\label{proof: fine-vs-coarse-partition}

The expected out-of-sample prediction error of the fine partition is
\[\sum_{t_i\in T_i} \mu(t_i) \ve_{t_i}^2 + \sigma_\eta^2\]
and the expected out-of-sample prediction error of the coarse partition is
\[\sum_{t_i\in T_i} \mu(t_i)\left(\sum_{t_i^\prime\in T_i} \mu(t_i^\prime) \ve_{t_i^\prime}  + \sum_{t_i^\prime\in T_i} \mu(t_i^\prime) \overline{x}_i^\sigma(t_i^\prime) -\overline{x}_i^\sigma(t_i)\right)^2 + \sigma_\eta^2.\]
Simple algebra shows that the fine partition is adopted if the inequality \eqref{eq: fine-vs-coarse-partition} holds.

\printbibliography

\end{document}